\documentclass[%
 reprint,
superscriptaddress,
 amsmath,amssymb,
aps,
pra,
showpacs
]{revtex4-2}

\usepackage{graphicx}
\usepackage{bm}
\usepackage{bbm}
\usepackage{amsthm,amsmath,amssymb}
\usepackage{color}
\usepackage{enumitem}
\usepackage{caption}
\usepackage{subcaption}
\usepackage{mathtools}
\usepackage{dcolumn}
\usepackage{hyperref}
\usepackage[ruled]{algorithm2e}
\usepackage{braket}
\allowdisplaybreaks[2]

\usepackage{tikz}
\usetikzlibrary{arrows.meta}

\theoremstyle{plain}
\newtheorem{theorem}{Theorem}[section]
\newtheorem{proposition}[theorem]{Proposition}

\theoremstyle{definition}

\newtheorem{Observation}{Observation}[section]
\theoremstyle{remark}

\newcommand{\cH}{\mathcal{H}}

\newcommand{\C}{\mathbb{C}}
\newcommand{\R}{\mathbb{R}}

\newcommand{\Tr}{\operatorname{Tr}}

\begin{document}


\title{Variance Geometry of Exact Pauli-Detecting Codes:\\ Continuous Landscapes Beyond Stabilizers}

\author{Arunaday Gupta}
\affiliation{Department of Physics, The University of Texas at Dallas, Richardson, Texas 75080, USA}
\author{Baisong Sun}
\affiliation{Department of Physics, The University of Texas at Dallas, Richardson, Texas 75080, USA}
\author{Xi He}
\affiliation{Department of Physics, The University of Texas at Dallas, Richardson, Texas 75080, USA}
\author{Bei Zeng}
\affiliation{Department of Physics, The University of Texas at Dallas, Richardson, Texas 75080, USA}

\begin{abstract}
Exact quantum codes detecting a prescribed set of Pauli errors are approached through algebraic constructions---stabilizer, codeword-stabilized, permutation-invariant, topological, and related families. Geometrically, exact Pauli detection is governed by joint higher-rank numerical ranges of these Pauli operators, whose structure for rank $\geq 2$ is largely uncharted. From this viewpoint, we show that such codes often form connected continuous families rather than collections of disjoint solution regions. These families are characterized by a single scalar derived from the Knill--Laflamme conditions: denoted $\lambda^*$, it is the Euclidean norm of the signature vector of Pauli expectation values on the maximally mixed code state, and provides a one-parameter summary of the code's joint Pauli variance profile. Within these continuous landscapes, stabilizer codes occupy only discrete, measure-zero subsets of the attainable $\lambda^*$-spectrum, exposing a largely unexplored continuum of genuinely nonadditive exact codes. We establish this picture by analyzing the geometry of higher-rank operator compressions, and extend it to symmetry-restricted settings where cyclic and permutation symmetries are imposed on both the error model and the code projector. Small-system cases reveal interval, singleton, and empty regimes through eigenvalue interlacing and symmetry-sector decompositions; larger systems are treated numerically via Stiefel-manifold optimization and symmetry-adapted parameterizations. In every unrestricted and symmetry-compatible case analyzed, the attainable $\lambda^*$-spectrum forms a single closed interval whenever nonempty---although a general proof remains open. These results place stabilizer, symmetric, and nonadditive code families within a unified higher-rank variance framework, suggesting a continuous geometric perspective on the landscape of exact quantum codes.
\end{abstract}

\maketitle

\section{Introduction}
\label{sec:intro}

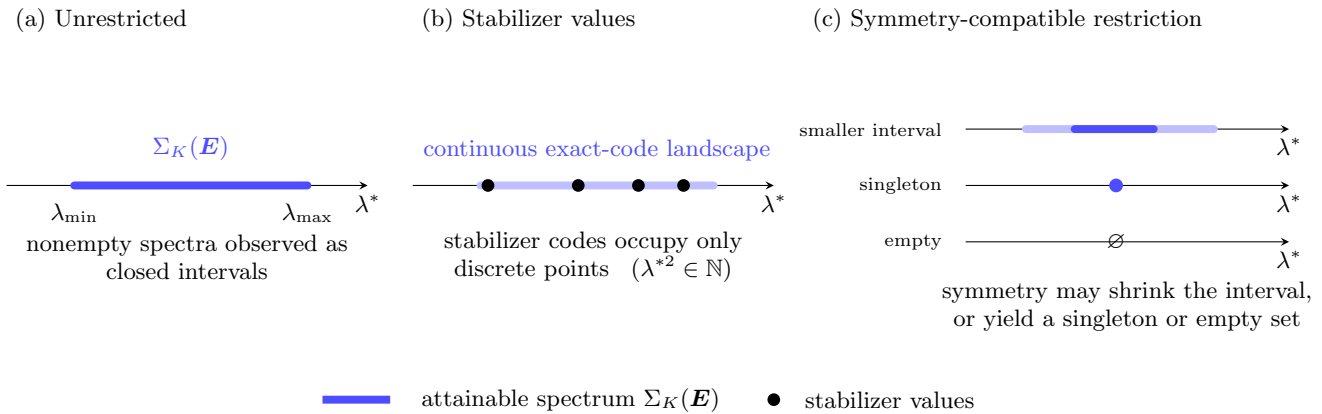
\begin{figure*}[t]
\centering
\begin{tikzpicture}[x=1cm,y=1cm,>=stealth,baseline=(current bounding box.center)]

\colorlet{specblue}{blue!70}
\colorlet{faintblue}{blue!25}

\begin{scope}[shift={(0,0)}]
  \node[anchor=west,font=\small] at (0,2.2) {(a) Unrestricted};
  \draw[->] (0,0) -- (4.8,0) node[below] {$\lambda^{*}$};
  \draw[line width=3pt,specblue,line cap=round] (0.9,0) -- (4.0,0);
  \node[below,font=\small] at (0.9,-0.12) {$\lambda_{\min}$};
  \node[below,font=\small] at (4.0,-0.12) {$\lambda_{\max}$};
  \node[above,font=\small,specblue] at (2.45,0.22) {$\Sigma_K(\boldsymbol E)$};
  \node[align=center,font=\small] at (2.4,-0.95)
    {nonempty spectra observed as\\closed intervals};
\end{scope}

\begin{scope}[shift={(5.4,0)}]
  \node[anchor=west,font=\small] at (0,2.2) {(b) Stabilizer values};
  \draw[->] (0,0) -- (4.8,0) node[below] {$\lambda^{*}$};
  \draw[line width=3pt,faintblue,line cap=round] (0.9,0) -- (4.0,0);
  \foreach \x in {1.0,2.2,3.0,3.6}
    \fill (\x,0) circle (2.4pt);
  \node[above,font=\small,specblue] at (2.45,0.22) {continuous exact-code landscape};
  \node[align=center,font=\small] at (2.4,-0.95)
    {stabilizer codes occupy only\\discrete points \;\;($\lambda^{*2}\in\mathbb N$)};
\end{scope}

\begin{scope}[shift={(12.2,0)}]
  \node[anchor=west,font=\small] at (-1.6,2.2) {(c) Symmetry-compatible restriction};

  \draw[->] (0.55,0.75) -- (4.85,0.75) node[below] {$\lambda^{*}$};
  \node[anchor=east,font=\scriptsize] at (0.35,0.75) {smaller interval};
  \draw[line width=3pt,faintblue,line cap=round] (1.35,0.75) -- (3.85,0.75);
  \draw[line width=3pt,specblue,line cap=round] (2.0,0.75) -- (3.05,0.75);

  \draw[->] (0.55,0) -- (4.85,0) node[below] {$\lambda^{*}$};
  \node[anchor=east,font=\scriptsize] at (0.35,0) {singleton};
  \fill[specblue] (2.55,0) circle (2.6pt);

  \draw[->] (0.55,-0.75) -- (4.85,-0.75) node[below] {$\lambda^{*}$};
  \node[anchor=east,font=\scriptsize] at (0.35,-0.75) {empty};
  \node[font=\small] at (2.55,-0.75) {$\varnothing$};

  \node[align=center,font=\small] at (2.7,-1.6)
    {symmetry may shrink the interval,\\
     or yield a singleton or empty set};
\end{scope}

\begin{scope}[shift={(4.2,-2.85)}]
  \draw[line width=3pt,specblue] (0,0) -- (0.9,0);
  \node[anchor=west,font=\small] at (1.2,0) {attainable spectrum \(\Sigma_K(\boldsymbol E)\)};
  \fill (6.0,0) circle (2.4pt);
  \node[anchor=west,font=\small] at (6.3,0) {stabilizer values};
\end{scope}

\end{tikzpicture}
\caption{Schematic behavior of the attainable spectrum
\(\Sigma_K(\boldsymbol E)=\{\lambda^*(P): P \text{ detects } \boldsymbol E,\ \operatorname{rank}(P)=K\}\).
In the unrestricted setting studied here, nonempty spectra are observed to be closed intervals. Stabilizer codes occupy only discrete values within these intervals. Under symmetry-compatible restrictions, the attainable set may shrink to a smaller interval, reduce to a singleton, or become empty, while remaining interval-like whenever nonempty.}
\label{fig:sigma-schematic}
\end{figure*}

Exact quantum codes that detect a prescribed family of Pauli errors are often approached through algebraic constructions---stabilizer codes \cite{CalderbankShor1996,Steane1996ErrorCorrectingCodes,Gottesman1997Stabilizer}, codeword-stabilized codes \cite{Cross2009CWS}, permutation-invariant codes \cite{Ouyang2014PermutationInvariant,Ouyang2016PermutationInvariantMultiQubit}, topological codes \cite{DennisKitaevLandahlPreskill2002}, and related families. From a geometric viewpoint \cite{bengtsson2017geometry}, however, a more primitive question comes first: once a Pauli error set $\mathcal E$ is specified, what is the size and shape of the space of exact detecting codes? The detectable set $\mathcal E$ may be the full family of Pauli operators up to a given weight, as in the usual distance-based setting, or a restricted subset motivated by biased or structured noise. Equivalently, one asks for the geometry of the rank-$K$ projectors satisfying the corresponding detection constraints. This landscape is carved out by simultaneous compression conditions on the nonconvex manifold of rank-$K$ code projectors, so there is no \emph{a priori} reason for it to admit a simple global form: one may encounter isolated solutions, disconnected components, or continuous families of different dimensions. Are the familiar stabilizer constructions representative of this landscape, or do they occupy only a structured subset inside a much larger continuum of nonadditive solutions \cite{Rains1997Nonadditive,Cross2009CWS,du2025klcoefficients}?

Since our focus is exact detection of a prescribed family of Pauli errors, we use the Knill--Laflamme condition in the detection form
\begin{equation}
P E P=\alpha_E P,
\qquad E\in\mathcal E,
\end{equation}
where $P$ is a rank-$K$ code projector. Equivalently, for any two states $|\psi\rangle$ and $|\phi\rangle$ in the code space,
$\langle \psi|E|\phi\rangle=\alpha_E\,\langle \psi|\phi\rangle$.
Thus the error carries no logical-state-dependent information inside the code: after compression back to the code, it acts only as a scalar. When $\alpha_E=0$, the error leaves the code space and is directly detected; when $\alpha_E\neq 0$, the error is degenerate but still logically indistinguishable within the code. When $\mathcal E$ is chosen to be the set of Pauli operators of weight at most $d-1$, one recovers the usual $((n,K,d))$ detection language \cite{KnillLaflamme}. The relation between this detection form and the broader correction form of the Knill--Laflamme conditions is reviewed in Sec.~\ref{sec:preliminaries}.

A projector satisfying this condition for every $E\in\mathcal E$ is therefore a common $K$-dimensional subspace on which an entire Pauli family compresses to scalars. Exact Pauli detection is thus a problem of simultaneous operator compression and links naturally to joint higher-rank numerical ranges, where one studies all scalar tuples obtainable by compressing several operators to a common $K$-dimensional subspace \cite{li2008generalized,ChoiKribsZyczkowski2006LAA,ChoiGiesingerHolbrookKribs2008,LiPoonSze2008JMAA,GauLiPoonSze2011,Woerdeman2008HigherRankConvex,afshin2018pauli}. Compared with rank-$1$ Pauli expectation geometry \cite{XuSchwonnekWinter2024,xu2025simultaneous} and with the better-understood higher-rank theory for individual operators \cite{ChoiKribsZyczkowski2006LAA,LiPoonSze2008JMAA,GauLiPoonSze2011,Woerdeman2008HigherRankConvex}, the joint higher-rank regime relevant here remains much less understood \cite{li2008generalized,afshin2018pauli}. In particular, for Pauli tuples and $K\ge 2$, there is no general structure theory predicting connectedness, interval behavior of natural scalar summaries, or sharp non-emptiness criteria. Exact Pauli detection therefore probes an operator-theoretic regime where any systematic regularity is mathematically interesting in its own right, not only for error correction.

For Hermitian Pauli observables, the same compression data has a direct variance interpretation, since $E^2=\mathbbm{1}$ implies
\begin{equation}
\mathrm{Var}_\rho(E)=1-\langle E\rangle_\rho^2.
\end{equation}
Thus each detecting projector $P$ determines, through the maximally mixed code state $\rho_P=P/K$, a joint profile of compressed Pauli expectations and variances on the code space. Fixing an ordered tuple $\boldsymbol{E}=(E_1,\dots,E_m)$ drawn from the detectable family, we define
\begin{equation}
\boldsymbol{\lambda}(P)=\bigl(\langle E_1\rangle_{\rho_P},\dots,\langle E_m\rangle_{\rho_P}\bigr),
\qquad
\lambda^*(P)=\|\boldsymbol{\lambda}(P)\|_2.
\end{equation}
The scalar $\lambda^*$ is therefore a compact, sign-insensitive summary of the code's simultaneous Pauli expectation profile, or equivalently its variance profile. In the exact-code setting of \cite{du2025klcoefficients}, the analogous scalar is invariant under local unitaries, and continuous families were found for the $((6,2,3))$ and $((7,2,3))$ code classes.

Once \(\lambda^*(P)\) is associated with an individual detecting projector, the natural global question is which values of \(\lambda^*\) are actually attainable. For a fixed detectable Pauli family \(\mathcal E\) and an ordered tuple \(\boldsymbol E=(E_1,\dots,E_m)\subseteq \mathcal E\), we define
\begin{equation}
\Sigma_K(\boldsymbol E)
:=
\left\{
\lambda^*(P):
\begin{array}{l}
\operatorname{rank}(P)=K,\; P~\text{detects}\\
\text{the prescribed Pauli set }\mathcal E
\end{array}
\right\}.
\end{equation}
We refer to \(\Sigma_K(\boldsymbol E)\) as the attainable scalar spectrum. While the full vector-valued feasible set can be complicated or even disconnected, \(\Sigma_K(\boldsymbol E)\) gives a one-dimensional order parameter that still captures the basic global questions of interest: feasibility, extremal values, and connectedness. It also provides a natural way to locate familiar algebraic subclasses inside the broader exact-code landscape.

From this perspective, stabilizer codes furnish a natural reference class, as illustrated schematically in Fig.~\ref{fig:sigma-schematic}(a)(b). Under a Pauli error model, each stabilizer signature coordinate is \(0\) or \(\pm1\), so stabilizer values of \(\lambda^*\) form a discrete set, equivalently \(\lambda^{*2}\in\mathbb N\). Thus, whenever \(\Sigma_K(\boldsymbol E)\) contains a nontrivial interval, stabilizer codes occupy only a discrete—and hence measure-zero within that interval—subset of the attainable exact-code landscape. The phenomenon observed here is stronger: in every Pauli error model studied, the entire attainable scalar spectrum \(\Sigma_K(\boldsymbol E)\) appears to be a single closed interval whenever it is nonempty. In this sense, stabilizer constructions appear not as isolated points in a fragmented solution set, but as discrete reference points embedded in a continuous interval of exact Pauli-detecting codes, most of which are nonadditive.

A further layer of structure comes from symmetry. Besides the choice of detectable Pauli family---full low-weight sets, biased families, or other channel-adapted restrictions \cite{IoffeMezard2007,Jackson2016CWSAsym}---one may impose cyclic or permutation symmetry on the code \cite{beigi_et_al:LIPIcs.TQC.2013.192,PollatsekRuskai2004,Ouyang2014PermutationInvariant,Ouyang2016PermutationInvariantMultiQubit}. Here it is crucial to distinguish symmetry-compatible detection problems, in which the detectable Pauli set is itself stable under the same group action, from externally imposed symmetry restrictions on an otherwise asymmetric tuple. It is equally important to distinguish symmetry of individual basis states from symmetry of the rank-$K$ projector. The schematic in Fig.~\ref{fig:sigma-schematic}(c) summarizes symmetry-compatible parts of this picture. Under symmetry-compatible restrictions, the same interval-like structure persists but may shrink in three possible ways: it may remain a smaller interval, collapse to a singleton, or become empty.  These correspond to the later \textbf{Case~1}--\textbf{Case~3} behaviors.  The contrasting \textbf{Case~4} behavior---a robust disconnected spectrum such as $\{0,1\}$---does not occur in the symmetry-compatible schematic of Fig.~\ref{fig:sigma-schematic}; it arises when symmetry is imposed externally on a Pauli tuple that is not stable under that symmetry, as shown in the random cyclic-$+1$ examples of Sec.~\ref{sec:n3_k2}.

The main message of this paper is that the scalar signature spectrum is far more regular than the underlying vector-valued compression geometry might suggest.  In every unrestricted Pauli-detecting problem studied here, $\Sigma_K(\boldsymbol E)$ is a single closed interval whenever it is nonempty, and the same interval behavior persists for all symmetry-compatible cyclic and permutation restrictions we analyze.  Symmetry can shrink this interval, collapse it to a singleton, or make it empty, but it does not break connectedness in the compatible setting.  By contrast, when symmetry is imposed externally on a Pauli tuple that is not stable under the symmetry action, this regularity can fail: random three-qubit examples exhibit disconnected spectra such as $\{0,1\}$.  A related theme is that the rank-$K$ projector, rather than a particular symmetric basis, is the natural symmetry carrier for exact detection; relaxing from state-level to projector-level symmetry can enlarge intervals or restore feasibility.  Thus $\lambda^*$ serves here not as an approximate-recovery benchmark, but as a geometric order parameter placing stabilizer, nonadditive, cyclic, permutation-invariant, and asymmetric exact codes within a common higher-rank variance framework.

The paper is organized as follows. Section~\ref{sec:preliminaries} fixes notation, reviews the operator-compression formulation of the Knill--Laflamme conditions, and recalls joint higher-rank numerical ranges. Section~\ref{sec:var} introduces the variance-based Pauli signatures, the signature norm \(\lambda^*\), and the attainable scalar spectrum \(\Sigma_K(\boldsymbol E)\), including the interval phenomenon studied throughout the paper. Section~\ref{sec:sym} develops the symmetry framework, distinguishing symmetry-compatible error models from externally imposed symmetry constraints and state-level from projector-level code symmetry. Section~\ref{sec:cases} gives the complete two-qubit rank-\(2\) classification, both without symmetry and under the swap constraint. Section~\ref{sec:opt} presents the Stiefel-manifold optimization framework and the symmetry-adapted parameterizations used in the numerical searches. Section~\ref{sec:n3_k2} analyzes the three-qubit \(n=3\), \(K=2\) setting, including unrestricted random Pauli tuples, externally imposed cyclic-\(+1\) restrictions with an exact disconnected example, and cyclic-stable Pauli families. Section~\ref{sec:state_vs_projector} extends the analysis to larger systems and examines cyclic and permutation state-versus-projector symmetry gaps. We conclude in Section~\ref{sec:discussion}.


\section{Preliminaries}
\label{sec:preliminaries}

This section fixes the notation and operator-compression viewpoint used throughout the paper. We first set conventions for qubit systems and code projectors, then recall joint higher-rank numerical ranges, and finally connect that language to the Knill--Laflamme conditions and stabilizer codes.

\subsection{Notation}

We work with an $n$-qubit Hilbert space $\mathcal H \cong (\mathbb C^2)^{\otimes n}$ of dimension $D=2^n$, and write $\mathcal B(\mathcal H)$ for the algebra of linear operators on $\mathcal H$. The $n$-qubit Pauli group is denoted by $\mathcal P_n$. When we speak of Pauli observables, we mean Hermitian elements of $\mathcal P_n$, namely tensor products of $I$, $X$, $Y$, and $Z$ up to an overall sign. For $E\in\mathcal P_n$, its weight $\mathrm{wt}(E)$ is the number of tensor factors different from $I$.

A quantum code is a subspace $\mathcal C\subset \mathcal H$ with orthogonal projector $P$, so $P=P^\dagger=P^2$ and $\mathcal C=\mathrm{Ran}(P)$. We write $K=\mathrm{rank}(P)$ for the code dimension and $\mathcal P_K$ for the set of rank-$K$ orthogonal projectors on $\mathcal H$; thus $\mathcal P_n$ refers to the Pauli group whereas $\mathcal P_K$ refers to projectors. The maximally mixed state supported on the code space is
\begin{equation}
\rho_P=\frac{P}{K}.
\end{equation}
For any state $\rho$ and operator $A$, we use the shorthand
\begin{equation}
\langle A\rangle_\rho:=\Tr(\rho A).
\end{equation}
The state $\rho_P$ will be used throughout as the basis-independent representative of the code. In the abstract numerical-range discussion below we use $A=(A_1,\dots,A_m)$ for a general Hermitian operator tuple. For Pauli operators, $\mathcal E$ denotes a finite error set to be detected, $\boldsymbol{E}=(E_1,\dots,E_m)$ a fixed ordered tuple drawn from~$\mathcal E$, and unsubscripted~$E$ a generic Hermitian element of~$\mathcal P_n$.

\subsection{Joint higher-rank numerical ranges}

For a single Hermitian operator $A\in\mathcal B(\mathcal H)$, its numerical range is
\begin{equation}
W(A)=\left\{
\langle \psi, A\psi\rangle :
\begin{array}{l}
\psi\in\mathcal H,\\
\|\psi\|=1
\end{array}
\right\}
\subseteq \mathbb R .
\end{equation}
namely the set of expectation values attainable on pure states. For an $m$-tuple of Hermitian operators $A=(A_1,\dots,A_m)$, the corresponding joint numerical range is
\begin{equation}
\Lambda_1(A)
=
\left\{
\left(
\begin{array}{c}
\langle \psi, A_1\psi\rangle,\dots, \\
\langle \psi, A_m\psi\rangle
\end{array}
\right)
:
\begin{array}{l}
\psi\in\mathcal H,\\
\|\psi\|=1
\end{array}
\right\}
\subseteq \mathbb R^m .
\end{equation}
To pass from vectors to $K$-dimensional subspaces, one asks for simultaneous scalar compressions of the operators onto a common rank-$K$ projector. This leads to the joint higher-rank numerical range
\begin{equation}
\Lambda_K(A)
=
\left\{
\lambda\in\mathbb R^m :
\begin{array}{l}
\exists\, P\in\mathcal P_K \text{ such that}\\
P A_j P=\lambda_j P,\quad j=1,\dots,m
\end{array}
\right\},
\label{eq:LambdaK}
\end{equation}
which reduces to the usual joint numerical range when $K=1$ \cite{afshin2018pauli, LiPoonSze2008JMAA, GauLiPoonSze2011, Woerdeman2008HigherRankConvex, LiSze2008ProcAMS, ChoiKribsZyczkowski2006LAA}. A point of $\Lambda_K(A)$ therefore encodes a $K$-dimensional subspace on which each $A_j$ acts as a scalar multiple of the identity. This compression viewpoint is the natural geometric language for the coding problems studied in the rest of the paper.

\subsection{Connection to quantum error detection and correction}

A quantum code is a rank-$K$ projector \(P\) on \(\mathcal H\), whose range is the encoded \(K\)-dimensional logical subspace. For a general error model, let \(\mathcal N=\{N_\alpha\}\subset\mathcal B(\mathcal H)\) be a finite set of error operators. Then \(P\) corrects \(\mathcal N\) precisely when there exists a Hermitian matrix \(c=(c_{\alpha\beta})\) such that
\begin{equation}
P N_\alpha^\dagger N_\beta P = c_{\alpha\beta} P,
\qquad \forall\, \alpha,\beta,
\label{eq:KL_correct}
\end{equation}
which is the Knill--Laflamme condition \cite{KnillLaflamme}. To express it in numerical-range language, choose Hermitian operators \(A_1,\dots,A_m\) spanning the real linear space generated by the Hermitian and skew-Hermitian parts of the products \(\{N_\alpha^\dagger N_\beta\}_{\alpha,\beta}\). Then \eqref{eq:KL_correct} holds if and only if
\begin{equation}
P A_j P=\lambda_j P,\qquad j=1,\dots,m
\end{equation}
for some real scalars \(\lambda_1,\dots,\lambda_m\), equivalently if and only if \((\lambda_1,\dots,\lambda_m)\in\Lambda_K(A)\). In this sense, the existence of a rank-$K$ quantum code correcting \(\mathcal N\) is equivalent to the non-emptiness of an appropriate joint higher-rank numerical range \cite{li2008generalized, du2025klcoefficients}. 

We now restrict to Pauli error models. For a Pauli family \(\mathcal E\subset\mathcal P_n\), the detection condition takes the simpler form
\begin{equation}
P E P=\alpha_E P,
\qquad \forall\, E\in\mathcal E.
\end{equation}
Whenever this holds, we call \(P\) an \(\mathcal E\)-detecting projector and \(\mathrm{Ran}(P)\) an \(\mathcal E\)-detecting code. In the standard distance-based setting, one takes
\begin{equation}
\mathcal E_d
:=
\{E\in\mathcal P_n:\mathrm{wt}(E)\le d-1\},
\label{eq:Ed-def}
\end{equation}
and says that \(P\) is an \(((n,K,d))\) code if it detects every Pauli operator in \(\mathcal E_d\). Thus the parameter \(d\) is simply a compact way of specifying the detectable Pauli set through a uniform weight cutoff.

For Pauli errors, the correction and detection equations are directly related. If a code corrects a Pauli set \(\mathcal E=\{E_\alpha\}\), then by \eqref{eq:KL_correct} it detects every operator in the product set
\begin{equation}
\mathcal E^\dagger \mathcal E
:=
\{E_\alpha^\dagger E_\beta:\ E_\alpha,E_\beta\in\mathcal E\},
\end{equation}
since each such product compresses to a scalar. Conversely, the Knill--Laflamme correction equations for \(\mathcal E\) are exactly the detection equations for \(\mathcal E^\dagger\mathcal E\). In particular, if \(t=\lfloor (d-1)/2\rfloor\), then correcting all Pauli errors of weight at most \(t\) is equivalent to detecting all Pauli operators of weight at most \(2t\); this is why an \(((n,K,d))\) code corrects arbitrary errors on up to \(\lfloor (d-1)/2\rfloor\) qubits.

In realistic devices, however, noise is often biased rather than depolarizing. Depending on the platform, mechanisms such as spontaneous emission, dephasing, or amplitude-damping-type processes can favor one Pauli direction over another; for example, phase-flip ($Z$) faults can dominate bit-flip ($X$) faults, and non-unitary damping effects can often be mapped to asymmetric Pauli channels with unequal $X$, $Y$, and $Z$ error rates \cite{IoffeMezard2007,Jackson2016CWSAsym}. This motivates \emph{asymmetric} or channel-adapted codes, by which we mean codes designed for a prescribed Pauli error set that is not simply the full family below a single weight threshold \cite{IoffeMezard2007,SarvepalliKlappeneckerRoetteler2008,FletcherShorWin2008,LangShor2007,ShorSmithSmolinZeng2011,EzermanGrassl2013}. In that setting, the usual distance parameter \(d\) no longer faithfully specifies the error model: Pauli operators of the same weight may be treated differently, and the relevant detectable set need not be determined by weight alone. For the asymmetric families studied in this paper, we specify the detectable Pauli set directly by augmenting the single-qubit Pauli family with successively higher-body Z-type correlators, thereby defining the \emph{hierarchical asymmetric family}
\begin{equation}
\begin{aligned}
\mathcal E^{\mathrm{asym}}_{n,r}
:=
&\{X_i,Y_i,Z_i\}_{i=1}^{n}\\
&\cup
\bigcup_{s=2}^{r}
\{Z_{i_1}Z_{i_2}\cdots Z_{i_s}\}_{1\le i_1<\cdots<i_s\le n},
\end{aligned}
\label{eq:asym-family}
\end{equation}
where $r\in\{2,3,\dots,n\}$ denotes the largest included body order. Thus $\mathcal E^{\mathrm{asym}}_{n,2}=\{X_i,Y_i,Z_i\}\cup\{Z_iZ_j\}_{i<j}$, and each increment of $r$ adds one further layer of $Z$-type strings. In addition, we will encounter mixed two-body families such as
\begin{equation}
\begin{aligned}
\mathcal E^{\mathrm{mix}}_{n}
:=
&\{X_i,Y_i,Z_i\}_{i=1}^{n}\\
&\cup
\{X_iX_j,\ Z_iZ_j,\ X_iZ_j,\ Z_iX_j\}_{1\le i<j\le n},
\end{aligned}
\label{eq:mix-family}
\end{equation}
which include both diagonal and off-diagonal two-body Pauli products.

Stabilizer codes provide a particularly structured family of such projectors. If \(\mathcal R\subset\mathcal P_n\) is an abelian subgroup that does not contain \(-\mathbbm{1}\), then the associated stabilizer code is the common \(+1\) eigenspace
\begin{equation}
\mathcal C
=
\{
|\psi\rangle\in\mathcal H :
R|\psi\rangle=|\psi\rangle,\ \forall\, R\in\mathcal R
\},
\end{equation}
with projector
\begin{equation}
P=\frac{1}{|\mathcal R|}\sum_{R\in\mathcal R}R.
\end{equation}
This projector has rank \(K=2^n/|\mathcal R|\) and satisfies \(PRP=P\) for all \(R\in\mathcal R\). For a Hermitian Pauli operator \(E\), the compression \(PEP\) is determined by its commutation relations with \(\mathcal R\): if \(E\) anticommutes with some element of \(\mathcal R\), then \(PEP=0\) and the error is detected; if \(E\) agrees with a stabilizer element up to sign, then \(PEP=\pm P\); and if \(E\) commutes with every stabilizer but is not equal to a stabilizer up to sign, then \(E\) lies in the normalizer outside the stabilizer and acts as a nontrivial logical operator, so it is not detectable. Stabilizer codes therefore appear in the present framework as distinguished higher-rank compression points whose Pauli compression coefficients lie in the discrete set \(\{0,\pm1\}\).

The compression formalism, however, is not limited to stabilizer projectors. Nonadditive codes satisfy the same Knill--Laflamme compression equations without arising from an abelian stabilizer subgroup, so in the present framework they appear as other feasible higher-rank compression points \cite{Rains1997Nonadditive, Cross2009CWS, Jackson2016CWSAsym}. From the viewpoint of \(\Lambda_K(A)\) and the variance-based quantities discussed later in Sec.~\ref{sec:var}, stabilizer codes should therefore be regarded not as exhausting the feasible geometry, but as a distinguished structured subset inside a generally larger landscape that can also contain genuinely nonadditive exact codes.

The geometry of the feasible higher-rank compression set acquires an additional layer of structure once symmetry restrictions are introduced. In particular, we will study symmetry-compatible detection problems in which cyclic or permutation symmetry preserves the detectable Pauli family and acts consistently on the code space \cite{beigi_et_al:LIPIcs.TQC.2013.192, Ouyang2014PermutationInvariant, Ouyang2016PermutationInvariantMultiQubit}. Within this framework, it will be important to distinguish symmetry of individual codewords from symmetry of the code space itself, or equivalently of its rank-\(K\) projector. This state-level versus projector-level distinction will be central in Sec.~\ref{sec:sym}.

\section{Variance-Based Signatures of Pauli Observables}
\label{sec:var}

In this section, we formulate a variance-based geometric framework for Pauli observables on error-detecting code spaces. Building on the compression perspective of Sec.~\ref{sec:preliminaries}, we express Pauli variances through squared expectation values, organize them into joint higher-rank variance ranges, and from these derive a scalar signature parameter $\lambda^*$, whose attainable values define the signature spectrum.

\subsection{Joint higher-rank variance ranges for Pauli observables}

While the joint higher-rank numerical range \(\Lambda_K\) records simultaneous scalar compressions of operators, our interest here is in the statistical data carried by Pauli observables on the maximally mixed code state \(\rho_P=P/K\).

For a Pauli observable \(E\in\mathcal P_n\), one has \(E^2=\mathbbm{1}\), and hence
\begin{equation}
\mathrm{Var}_\rho(E)
=
\langle E^2\rangle_\rho-\langle E\rangle_\rho^2
=
1-\langle E\rangle_\rho^2 .
\end{equation}
Thus, for Pauli observables, specifying the variance is equivalent to specifying the squared expectation value. The latter is invariant under the sign change \(E\mapsto -E\), so it provides a natural sign-insensitive coordinate for the geometry we wish to study.

For a single Pauli observable \(E\), we write
\begin{equation}
V(E)
:=
\left\{
\langle E\rangle_\rho^2
:\;
\rho\in\mathcal D(\mathcal H)
\right\},
\end{equation}
\begin{equation}
\mathcal D(\mathcal H)
=
\{\rho\in\mathcal B(\mathcal H):\rho\ge0,\ \Tr\rho=1\},
\end{equation}
and refer to this set as the variance range of \(E\). Since a Pauli observable has eigenvalues \(\pm1\), every value in \([0,1]\) is attainable, and therefore \(V(E)=[0,1]\). The nontrivial structure only appears after imposing higher-rank compression and error-detection constraints.

Motivated by recent work on Pauli expectation and variance geometry
\cite{xu2025simultaneous, Cabello2014Graph, XuSchwonnekWinter2024, BorelandTodorovWinter2022ConvexCorners, BorelandTodorovWinter2021, DuanSeveriniWinter2013},
we now fix a finite Pauli error set \(\mathcal E\subset\mathcal P_n\) and an ordered tuple \(\boldsymbol{E}=(E_1,\dots,E_m)\) selected from \(\mathcal E\). If \(P\in\mathcal P_K\) detects all operators in \(\mathcal E\), then in particular each \(E_\alpha\) compresses to a scalar on the code space:
\begin{equation}
P E_\alpha P=\lambda_\alpha(P)\,P,
\qquad \alpha=1,\dots,m .
\end{equation}
Since \(\rho_P=P/K\), these compression coefficients are exactly the corresponding code-state expectations,
\begin{equation}
\lambda_\alpha(P)
=
\Tr(\rho_P E_\alpha)
=
\langle E_\alpha\rangle_{\rho_P}.
\end{equation}
We therefore associate with \(P\) the variance profile
\begin{equation}
q(P)
:=
\bigl(
q_1(P),\dots,q_m(P)
\bigr),
\end{equation}
\begin{equation}
q_\alpha(P)
:=
\lambda_\alpha(P)^2
=
\langle E_\alpha\rangle_{\rho_P}^2,
\label{eq:variance-profile}
\end{equation}
and define the joint rank-$K$ variance range by
\begin{equation}
Q_K(\boldsymbol{E})
:=
\left\{
q(P)\in[0,1]^m
:\;
P\in\mathcal P_K \text{ is }\mathcal E\text{-detecting}
\right\},
\label{eq:QK-def}
\end{equation}
with the ambient detectable set \(\mathcal E\) understood from context. Equivalently, \(Q_K(\boldsymbol{E})\) is the set of all squared compression-coordinate vectors arising from rank-$K$ projectors that detect the prescribed Pauli set.

From this perspective, \(Q_K(\boldsymbol{E})\) is a higher-rank, code-constrained analogue of the joint variance body. Unlike the rank-$1$ setting, however, the admissible points are filtered by the requirement that one and the same rank-$K$ projector satisfy all of the Knill--Laflamme detection constraints. As a result, \(Q_K(\boldsymbol{E})\) is generally neither convex nor closed under arbitrary convex combinations, because mixtures of the maximally mixed code states \(\rho_P\) need not arise from a single rank-$K$ detecting projector.

\subsection{Signature vector and the parameter \texorpdfstring{$\lambda^*$}{lambda*}}

For each rank-$K$ projector \(P\) detecting \(\mathcal E\), the scalar compression coefficients naturally assemble into the signature vector
\begin{equation}
\boldsymbol{\lambda}(P)
:=
\bigl(\lambda_\alpha(P)\bigr)_{\alpha=1}^m
=
\bigl(
\langle E_\alpha\rangle_{\rho_P}
\bigr)_{\alpha=1}^m
\in \mathbb R^m .
\end{equation}
This vector records the simultaneous Pauli expectations visible on the code space. In general its admissible set can inherit the discontinuities and degeneracies of higher-rank compression problems. For example, for two qubits with Pauli observables
\(
E_1=X\otimes I,\;
E_2=X\otimes Z,\;
E_3=Y\otimes I,\;
E_4=Y\otimes Z,\;
E_5=Z\otimes I,
\)
the rank-$2$ compression admits only the two expectation profiles
\((0,0,0,0,1)\) and \((0,0,0,0,-1)\) \cite{du2025klcoefficients}, so the admissible signature vectors form a disconnected set.

To extract a scalar summary that is insensitive to sign changes and permutations of the coordinates, we define the signature norm by
\begin{equation}
\lambda^*(P)
:=
\|\boldsymbol{\lambda}(P)\|_2
=
\left(
\sum_{\alpha=1}^m
\langle E_\alpha\rangle_{\rho_P}^2
\right)^{1/2}.
\label{eq:lambda-star}
\end{equation}
Combining \eqref{eq:variance-profile} and \eqref{eq:lambda-star} gives
\begin{equation}
\lambda^{*2}(P)=\sum_{\alpha=1}^m q_\alpha(P).
\end{equation}
Hence \(\lambda^{*2}\) is exactly the image of \(Q_K(\boldsymbol{E})\) under the linear functional \(q\mapsto \sum_\alpha q_\alpha\). The quantity \(\lambda^*\) itself is the corresponding Euclidean radius, providing an isotropic one-dimensional summary of the variance geometry that is particularly convenient for analysis and numerical optimization.

\subsection{Problem statement and the interval phenomenon}

With \(\mathcal E\) and the tuple \(\boldsymbol{E}=(E_1,\dots,E_m)\) fixed, the set of all attainable signature norms is
\begin{equation}
\Sigma_K(\boldsymbol{E})
:=
\left\{
\lambda^*(P)
:\;
P\in\mathcal P_K \text{ is }\mathcal E\text{-detecting}
\right\},
\end{equation}
again with the background detectable set \(\mathcal E\) understood from context. Since \(|\lambda_\alpha(P)|\le1\) for every \(\alpha\), one always has
\begin{equation}
\Sigma_K(\boldsymbol{E})\subseteq[0,\sqrt m].
\end{equation}
Whenever \(Q_K(\boldsymbol{E})\) is non-empty, the spectrum \(\Sigma_K(\boldsymbol{E})\) gives a one-dimensional summary of the underlying higher-rank variance geometry. Our aim is to understand its extremal values
\(\lambda_{\min}:=\inf\Sigma_K(\boldsymbol{E})\) and
\(\lambda_{\max}:=\sup\Sigma_K(\boldsymbol{E})\), together with the overall shape of the set.

Stabilizer codes provide a useful discrete reference class. As recalled in Sec.~\ref{sec:preliminaries}, if \(P\) is a stabilizer projector and \(E_\alpha\) is a Hermitian Pauli operator that compresses scalarly on the code space, then the corresponding coefficient \(\lambda_\alpha(P)\) belongs to \(\{0,\pm1\}\): it is \(0\) for detected Paulis and \(\pm1\) when \(E_\alpha\) agrees with a stabilizer element up to sign. Consequently, every stabilizer signature vector has coordinates in \(\{0,\pm1\}\), and the associated signature norm belongs to a discrete subset of \([0,\sqrt m]\). Stabilizer codes therefore serve as natural reference points for the extremal structure of \(\Sigma_K(\boldsymbol{E})\).

A second structural feature comes from Clifford symmetry. Let \(U\) be an \(n\)-qubit Clifford unitary and set \(P'=UPU^\dagger\). Then
\begin{equation}
\lambda_\alpha(P')
=
\Tr(\rho_{P'}E_\alpha)
=
\Tr(\rho_P U^\dagger E_\alpha U).
\end{equation}
If the Clifford action preserves the chosen tuple up to signs and permutations, so that
\begin{equation}
U^\dagger E_\alpha U=\pm E_{\pi(\alpha)}
\end{equation}
for some permutation \(\pi\), then the signature vector transforms by a signed permutation,
\begin{equation}
\boldsymbol{\lambda}(P')=D\Pi\,\boldsymbol{\lambda}(P),
\end{equation}
where \(\Pi\) is the permutation matrix of \(\pi\) and \(D\) is a diagonal matrix with entries \(\pm1\). In particular, \(\lambda^*(P')=\lambda^*(P)\). Moreover, if the full detectable set \(\mathcal E\) is invariant under the same Clifford action, then Clifford conjugation preserves the admissible family of rank-$K$ projectors and hence leaves \(\Sigma_K(\boldsymbol{E})\) unchanged.

Extensive numerical computations for small system sizes and a variety of Pauli sets reveal a striking regularity, which we record as an empirical observation.

\begin{Observation}[Interval property for Pauli signature spectra]
\label{ob:interval}
Fix a finite Pauli set \(\mathcal E\), an ordered tuple \(\boldsymbol{E}=(E_1,\dots,E_m)\) selected from \(\mathcal E\), and an integer \(K\) for which \(Q_K(\boldsymbol{E})\) is non-empty. In all cases tested, the signature spectrum \(\Sigma_K(\boldsymbol{E})\) is a single closed interval,
\begin{equation}
\Sigma_K(\boldsymbol{E})=[\lambda_{\min},\lambda_{\max}]\subseteq[0,\sqrt m].
\end{equation}
\end{Observation}

Concrete analytic examples of this phenomenon already appear in the distance-$3$ families studied in Ref.~\cite{du2025klcoefficients}.  For the $((6,2,3))$ family, the attainable signature norm was shown to fill the interval
$\lambda^*\in\left[\sqrt{0.6},\,1\right]$,
with the endpoint $\lambda^*=1$ realized by a degenerate stabilizer code and the interior realized by continuous families of nonadditive codes.  For the $((7,2,3))$ family, the corresponding interval is
$\lambda^*\in[0,\sqrt7]$,
where the lower endpoint is realized by the nondegenerate Steane code \cite{Steane1996ErrorCorrectingCodes}, the upper endpoint by the permutation-invariant Pollatsek--Ruskai code \cite{PollatsekRuskai2004}, and intermediate values by cyclic nonadditive codes.  Equivalently, in these distance-$3$ examples one is correcting single-qubit Pauli errors, or detecting the Pauli products of weight at most two that appear in the Knill--Laflamme products.  These two families provide concrete analytic instances of interval spectra, but the evidence developed below indicates that the phenomenon is much more general: for Pauli-detecting problems, including random and asymmetric Pauli error models, the attainable scalar spectrum appears to remain a single closed interval whenever it is nonempty.  Thus the interval behavior is not merely a special feature of these two distance-$3$ code families, but seems to reflect a broader regularity of exact Pauli-detecting code geometry.

At present, Observation~\ref{ob:interval} concerns only the one-dimensional scalar summary extracted from \(Q_K(\boldsymbol{E})\) through the map
\begin{equation}
q\longmapsto \left(\sum_{\alpha=1}^m q_\alpha\right)^{1/2}.
\end{equation}
In many examples, the full set \(Q_K(\boldsymbol{E})\) itself also appears to be connected inside \([0,1]^m\), but a systematic study of that higher-dimensional connectedness problem is beyond the scope of this paper. For the purposes of the present work, \(\lambda^*\) will therefore serve as the primary organizing parameter.


\section{Symmetry Restrictions}
\label{sec:sym}

With \(\lambda^*\) established as the organizing parameter of the signature spectrum, the next natural question is how symmetry constraints on the code space influence the attainable range of \(\lambda^*\). The symmetry considerations in this paper are most naturally formulated not by imposing cyclic or permutation invariance on an otherwise arbitrary detection problem, but by restricting attention from the outset to \emph{symmetry-compatible} detection problems, namely those for which the detectable Pauli family is itself stable under the same qubit-relabeling action. As discussed in Sec.~\ref{sec:preliminaries}, the detectable Pauli families considered here fall into two broad classes. The first is the standard distance-based weight-bounded family \(\mathcal E_d\) defined in \eqref{eq:Ed-def}. Because Pauli weight is preserved under qubit permutations, \(\mathcal E_d\) is invariant under the full symmetric group, and therefore under any subgroup, including the cyclic group. The second class consists of asymmetric or channel-adapted families specified directly rather than through a single parameter \(d\), as in \eqref{eq:asym-family} and \eqref{eq:mix-family}. Although such families are not determined by a uniform weight threshold, they are still organized into symmetry orbits and remain stable under the relevant group action.

Within this symmetry-compatible framework, we consider two symmetry settings---cyclic and permutation symmetry---and, for each, distinguish two ways of imposing symmetry on the code space: a weaker \emph{projector-level} condition, requiring invariance of the code space as a subspace, and a stronger \emph{state-level} condition, requiring each codeword to be individually fixed by the relevant group action. In all symmetry-compatible cases studied here, however, the interval phenomenon persists: whenever the attainable \(\lambda^*\)-spectrum is nonempty, it forms a single closed interval.

\subsection{Symmetry constraints on \texorpdfstring{rank-$K$}{rank-K} code projectors}

Let \(G\) be a finite group acting unitarily on \(\mathcal H\) through a representation
\begin{equation}
g\longmapsto U_g .
\end{equation}
We call the pair \((\mathcal E,G)\) \emph{symmetry-compatible} when the detectable Pauli set is \(G\)-stable, meaning
\begin{equation}
U_g\,\mathcal E\,U_g^\dagger=\mathcal E,
\qquad \forall\,g\in G.
\end{equation}
A rank-$K$ projector \(P\in\mathcal P_K\) will be called \(G\)-symmetric when
\begin{equation}
U_g P U_g^\dagger=P,
\qquad \forall\,g\in G,
\label{eq:G-inv}
\end{equation}
or equivalently when \(P\) commutes with every \(U_g\). This condition requires only that the code space \(\operatorname{Ran}(P)\) be invariant under the group action; it does not require individual codewords to be fixed by the symmetry. The usefulness of \eqref{eq:G-inv} is immediate at the level of detection constraints. If \(P\) is both \(G\)-symmetric and \(\mathcal E\)-detecting, then for every \(E\in\mathcal E\) and every \(g\in G\),
\begin{equation}
P(U_g E U_g^\dagger)P
=
U_g(PEP)U_g^\dagger
=
\alpha_E P.
\end{equation}
Hence all Pauli operators belonging to the same \(G\)-orbit have the same compression coefficient,
\begin{equation}
\alpha_{U_g E U_g^\dagger}=\alpha_E,
\end{equation}
so the detection equations need only be enforced on orbit representatives. When the chosen tuple \(\boldsymbol{E}=(E_1,\dots,E_m)\) is likewise preserved up to signs and permutations, the signature norm \(\lambda^*\) from Sec.~\ref{sec:var} is compatible with the same symmetry.

If \(\mathcal E\) is not \(G\)-stable, one may still intersect the family of \(\mathcal E\)-detecting projectors with the \(G\)-symmetric ones. In that case, however, the orbit-reduction interpretation described above no longer applies, and the problem is better understood as one with an externally imposed symmetry constraint rather than a genuinely symmetry-adapted setting.

The representation-theoretic content of \eqref{eq:G-inv} is most transparent after decomposing \(\mathcal H\) into isotypic components. Choose a decomposition
\begin{equation}
\mathcal H
\cong
\bigoplus_{\mu}
\mathcal M_\mu\otimes \mathcal R_\mu,
\qquad
U_g
=
\bigoplus_{\mu}
\left(
I_{\mathcal M_\mu}\otimes r_\mu(g)
\right),
\end{equation}
where each \(\mathcal R_\mu\) carries an irreducible representation \(r_\mu\) of \(G\), and \(\mathcal M_\mu\) is the corresponding multiplicity space. By Schur's lemma \cite{Serre1977LinearRep}, a projector commutes with the full \(G\)-action if and only if it has the form
\begin{equation}
P
=
\bigoplus_{\mu}
\left(
Q_\mu\otimes I_{\mathcal R_\mu}
\right),
\qquad
Q_\mu^\dagger=Q_\mu^2=Q_\mu,
\end{equation}
with rank constraint
\begin{equation}
\sum_{\mu}\operatorname{rank}(Q_\mu)\,\dim\mathcal R_\mu=K.
\end{equation}
Thus symmetry reduces the search for \(P\) to a family of projectors on multiplicity spaces. This makes precise the distinction between projector-level and state-level symmetry: projector-level symmetry only asks that the code subspace be invariant, whereas state-level symmetry would confine every codeword to the trivial invariant sector.

For the fixed tuple \(\boldsymbol{E}=(E_1,\dots,E_m)\subset\mathcal E\), we therefore consider the symmetry-restricted signature spectrum
\begin{equation}
\Sigma_K^{(G)}(\boldsymbol{E})
:=
\left\{
\lambda^*(P)
:
\begin{array}{l}
P\in\mathcal P_K \text{ is }\mathcal E\text{-detecting}\\ 
\text{and}~G\text{-symmetric}
\end{array}
\right\},
\end{equation}
which is a subset of the unrestricted spectrum \(\Sigma_K(\boldsymbol{E})\). Extensive computations indicate that the interval phenomenon from Sec.~\ref{sec:var} persists under symmetry restriction for the symmetry-compatible error sets considered here.

\begin{Observation}[Interval property under symmetry restriction]
\label{ob:interval-sym}
In all cases tested, whenever \(\Sigma_K^{(G)}(\boldsymbol{E})\) is non-empty, it is a closed interval,
\begin{equation}
\Sigma_K^{(G)}(\boldsymbol{E})
=
[\lambda_{\min}^{(G)},\lambda_{\max}^{(G)}]
\subseteq[0,\sqrt m].
\end{equation}
\end{Observation}

Thus, within a symmetry-compatible coding problem, imposing symmetry typically shrinks or preserves the interval endpoints but does not break the interval into disconnected pieces.

\subsection{Cyclic symmetry}

Let \(T\) denote the unitary implementing a one-step cyclic shift of qubits,
\begin{equation}
T\,|b_1 b_2 \cdots b_n\rangle
=
|b_n b_1 \cdots b_{n-1}\rangle,
\qquad b_i\in\{0,1\},
\end{equation}
so that \(T^n=\mathbbm{1}\). The cyclic group \(C_n=\langle T\rangle\cong\mathbb Z_n\) acts unitarily on \(\mathcal H\), and a rank-$K$ projector \(P\) is cyclic precisely when
\begin{equation}
TPT^\dagger=P,
\end{equation}
or equivalently when \([P,T]=0\).

Because \(C_n\) is abelian, all of its irreducible representations are one-dimensional. Writing \(\omega=e^{2\pi i/n}\), the corresponding character sectors are
\begin{equation}
\mathcal H_\ell
:=
\{\,|\psi\rangle\in\mathcal H : T|\psi\rangle=\omega^\ell |\psi\rangle\,\},
\qquad
\ell=0,1,\dots,n-1,
\end{equation}
and they yield the decomposition
\begin{equation}
\mathcal H=\bigoplus_{\ell=0}^{n-1}\mathcal H_\ell.
\end{equation}
The projector onto \(\mathcal H_\ell\) is the Fourier projector
\begin{equation}
\Pi_\ell
=
\frac{1}{n}\sum_{t=0}^{n-1}\omega^{-\ell t}T^t.
\label{eq:cyclic-sector-proj}
\end{equation}
A projector is cyclic if and only if it is block diagonal with respect to this decomposition, equivalently
\begin{equation}
P=\sum_{\ell=0}^{n-1}\Pi_\ell P \Pi_\ell.
\end{equation}
Hence a cyclic projector may contain an arbitrary rank-$k_\ell$ subprojector inside each sector \(\mathcal H_\ell\), provided that \(\sum_{\ell=0}^{n-1}k_\ell=K\). In particular, projector-level cyclic symmetry still leaves substantial internal freedom whenever the sectors \(\mathcal H_\ell\) have dimension larger than one.

Literal state-level cyclic invariance is more restrictive. Requiring each codeword to satisfy
\begin{equation}
T|\psi\rangle=|\psi\rangle
\end{equation}
confines the code to the trivial character sector \(\mathcal H_0\). Projector-level cyclic symmetry, by contrast, only asks that the code space be invariant as a whole, so the projector may have support on one or several character sectors. This distinction will be important in the analytic examples and in the symmetry-adapted optimization schemes later in the paper.

\subsection{Permutation invariance}

For the full symmetric group \(S_n\), let \(U_\pi\) denote the unitary that permutes tensor factors according to \(\pi\in S_n\). A rank-$K$ projector \(P\in\mathcal P_K\) is permutation-invariant precisely when
\begin{equation}
U_\pi P U_\pi^\dagger=P,
\qquad \forall\,\pi\in S_n.
\end{equation}

The \(S_n\)-action admits a richer decomposition because \(S_n\) is non-abelian. For \(n\) qubits, Schur--Weyl duality gives \cite{FultonHarris1991}
\begin{equation}
\mathcal H
\cong
\bigoplus_{\lambda\vdash n,\ \text{at most two rows}}
\mathcal M_\lambda\otimes \mathcal R_\lambda,
\end{equation}
\begin{equation}
U_\pi
=
\bigoplus_{\lambda}
\left(
I_{\mathcal M_\lambda}\otimes r_\lambda(\pi)
\right),
\end{equation}
where \(\mathcal R_\lambda\) is the irreducible \(S_n\)-module labeled by the partition \(\lambda\), and \(\mathcal M_\lambda\) is the corresponding multiplicity space. Accordingly, a permutation-invariant projector has the commutant form
\begin{equation}
P
=
\bigoplus_{\lambda}
\left(
Q_\lambda\otimes I_{\mathcal R_\lambda}
\right),
\end{equation}
where each \(Q_\lambda\) is a projector on \(\mathcal M_\lambda\) and
\begin{equation}
\sum_{\lambda}\operatorname{rank}(Q_\lambda)\,\dim\mathcal R_\lambda=K.
\end{equation}
Projector-level permutation symmetry therefore permits code spaces supported on nontrivial \(S_n\) sectors, as long as they are assembled in this representation-theoretic way.

The most familiar special case is the trivial representation sector, namely the fully symmetric subspace
\begin{equation}
\mathcal H_{\mathrm{sym}}
=
\{\,|\psi\rangle\in\mathcal H : U_\pi|\psi\rangle=|\psi\rangle\ \forall\,\pi\in S_n\,\}.
\end{equation}
This subspace is spanned by the Dicke states \cite{Dicke1954}
\begin{equation}
|D_{n,w}\rangle
:=
\binom{n}{w}^{-1/2}
\sum_{\substack{x\in\{0,1\}^n\\ |x|=w}}
|x\rangle,
\qquad
w=0,1,\dots,n,
\end{equation}
so \(\dim\mathcal H_{\mathrm{sym}}=n+1\). Any projector supported on \(\mathcal H_{\mathrm{sym}}\) is automatically permutation-invariant, and many explicit code constructions exploit exactly this sector.

However, this familiar state-level model does not exhaust projector-level permutation invariance. Requiring each codeword to be permutation-invariant forces the code into \(\mathcal H_{\mathrm{sym}}\), whereas requiring only
\begin{equation}
U_\pi P U_\pi^\dagger=P
\end{equation}
allows the code space to occupy additional nontrivial \(S_n\) sectors. The latter condition is strictly weaker and, in practice, often yields a much larger feasible family of \(\mathcal E\)-detecting codes. We will keep these two notions separate throughout the remainder of the paper.

\section{Classification for the Two-qubit Case}
\label{sec:cases}

We now illustrate the geometric mechanisms behind the signature spectrum through explicit classification of two-qubit case. In these low-dimensional settings, the Knill--Laflamme compression conditions $P E P=\alpha_E P$ can often be analyzed exactly, allowing us to identify when $\Sigma_K(\boldsymbol{E})$ forms a continuous interval and when symmetry constraints instead produce rigidity or nonexistence. These examples also prepare the ground for the more subtle existence gaps between state-level and projector-level symmetries discussed in Section~\ref{sec:state_vs_projector}.


We begin with the smallest nontrivial setting: two qubits ($n=2$) and a rank-$2$ code projector ($K=2$). Throughout, $\mathcal{H}\cong(\mathbb{C}^2)^{\otimes 2}$ with computational basis $\{\lvert 00\rangle,\lvert 01\rangle,\lvert 10\rangle,\lvert 11\rangle\}.$ We use the shorthand
\begin{equation}
\begin{aligned}
X_1&:=X\otimes I, & Y_1&:=Y\otimes I, & Z_1&:=Z\otimes I,\\
X_2&:=I\otimes X, & Y_2&:=I\otimes Y, & Z_2&:=I\otimes Z.
\label{eq:twoqubit_paulis}
\end{aligned}
\end{equation}
The projector $P$ \emph{detects} $\mathcal{E}$ iff
\begin{equation}
P E P=\alpha_E P, \qquad \forall\,E\in\mathcal{E},
\label{eq:twoqubit_KL}
\end{equation}
for scalars $\alpha_E\in\mathbb{R}$, and the corresponding signature coordinates are $\lambda_E(P)=\Tr(\rho_P E)=\alpha_E$ with $\rho_P=P/2$.

\subsection{Unconstrained classification.}
We first analyze the problem without imposing any symmetry constraint on the code space. In this unconstrained setting the Knill--Laflamme conditions are invariant under global unitary conjugation: if $(\mathcal{E},P)$ satisfies~\eqref{eq:twoqubit_KL}, then so does $(U\mathcal{E} U^\dagger ,\, U P U^\dagger)$. Since the Pauli operators are closed under Clifford conjugation, the signature spectrum depends only on the Clifford equivalence class of $\mathcal{E}$, and solutions may therefore be organized up to two-qubit Clifford operations and qubit relabeling.

Exhaustive enumeration over Pauli subsets $\mathcal{E}\subset \mathcal{P}_2\setminus\{I\}$, organized up to two-qubit Clifford conjugation and qubit relabeling, shows that the admissible instances are precisely those for which $Q_2(\boldsymbol{E})\neq\emptyset$, and that these fall into a finite collection of equivalence classes. For every such admissible $\mathcal{E}$ the signature spectrum takes one of the three forms
\begin{equation}
\Sigma_2(\boldsymbol{E}) \in \bigl\{ \{0\},\ \{1\},\ [0,1] \bigr\},
\label{eq:n2_three_forms}
\end{equation}
and hence is always a single closed interval, in agreement with Observation~\ref{ob:interval}.

These three possibilities are fully consistent with our numerical optimization. Classifying all $35$ representative Pauli error sets $\mathcal{E}$ in the unconstrained two-qubit problem yields exactly these three categories. The corresponding sets are listed below:

\begin{widetext}
\raggedright
\noindent\textbf{Continuous Interval Group ($\Sigma_2(\boldsymbol{E}) = [0,1]$):} \\
$\{IX\}$, $\{IX, IY\}$, $\{IX, IY, XX\}$, $\{IX, IY, XZ\}$, $\{IX, IY, XZ, YZ\}$, $\{IX, IY, XI, YZ\}$, $\{IX, IY, XI, XZ, YZ\}$.

\vspace{0.5em}
\noindent\textbf{Singleton Zero Group ($\Sigma_2(\boldsymbol{E}) = \{0\}$):} \\
$\{IX, XI\}$, $\{IX, IY, XI\}$, $\{IX, IY, XI, XZ\}$, $\{IX, IY, XI, YI\}$, $\{IX, IY, XI, YX\}$, $\{IX, IY, XX, XY\}$, $\{IX, IY, XI, XZ, YI\}$, $\{IX, IY, XI, YX, YY\}$, $\{IX, IY, XI, XZ, YI, YZ\}$, $\{IX, IY, XI, XZ, YI, ZX\}$, $\{IX, IY, XI, YX, YY, ZZ\}$, $\{IX, IY, XI, XZ, YI, YZ, ZX\}$, $\{IX, IY, XI, XZ, YI, YZ, ZX, ZY\}$.

\vspace{0.5em}
\noindent\textbf{Singleton One Group ($\Sigma_2(\boldsymbol{E}) = \{1\}$):} \\
$\{IX, IY, IZ\}$, $\{IX, IY, IZ, XX\}$, $\{IX, IY, IZ, XX, XY\}$, $\{IX, IY, IZ, XX, YX\}$, $\{IX, IY, IZ, XX, YY\}$, $\{IX, IY, XI, YI, ZZ\}$, $\{IX, IY, XZ, YZ, ZZ\}$, $\{IX, IY, IZ, XX, XY, YX\}$, $\{IX, IY, IZ, XX, YX, ZX\}$, $\{IX, IY, IZ, XX, YX, ZY\}$, $\{IX, IY, IZ, XX, XY, YX, YY\}$, $\{IX, IY, IZ, XX, XY, YX, ZX\}$, $\{IX, IY, IZ, XX, XY, YX, ZY\}$, $\{IX, IY, IZ, XX, XY, YX, YY, ZX\}$, $\{IX, IY, IZ, XX, XY, YX, YY, ZX, ZY\}$.
\end{widetext}

We now describe analytic mechanisms behind the three possibilities in \eqref{eq:n2_three_forms}. For continuous intervals we exhibit parametric code families whose signature sweeps $[0,1]$; for the two singleton cases we show that the Knill--Laflamme conditions force a unique value of $\lambda^*$.

\subsubsection{Continuous interval.}
Consider the rank-$2$ family of codewords
\begin{equation}
\begin{aligned}
\lvert 0_L\rangle &= a\,\lvert 00\rangle + b\,\lvert 01\rangle, \\
\lvert 1_L\rangle &= X_1 \,\lvert 0_L\rangle
= a\,\lvert 10\rangle + b\,\lvert 11\rangle,
\label{eq:n2_interval_family_states}
\end{aligned}
\end{equation}
where $a\in\mathbb{R}$, $b\in\mathbb{C}$, and $|a|^2+|b|^2=1$. Define
\begin{equation}
P(a,b):=\lvert 0_L\rangle\!\langle 0_L\rvert+\lvert 1_L\rangle\!\langle 1_L\rvert.
\label{eq:n2_interval_family_proj}
\end{equation}
A direct computation gives the scalar compressions
\begin{equation}
\begin{aligned}
P(a,b)\,X_2\,P(a,b) &= 2\,\Re(\bar a b)\,P(a,b), \\
P(a,b)\,Y_2\,P(a,b) &= 2\,\Im(\bar a b)\,P(a,b).
\label{eq:n2_interval_family_compressions}
\end{aligned}
\end{equation}
Consequently, for the single-observable tuple $(X_2)$ one has
$\lambda^*(P(a,b))=\bigl|2\,\Re(\bar a b)\bigr|$, while for the two-observable tuple $(X_2,Y_2)$
\begin{equation}
\lambda^*(P(a,b))
=\sqrt{\bigl(2\,\Re(\bar a b)\bigr)^2+\bigl(2\,\Im(\bar a b)\bigr)^2}
=2|ab|.
\label{eq:n2_interval_family_lambda}
\end{equation}
As $(a,b)$ ranges over the unit sphere, $2|ab|$ fills the full interval $[0,1]$. Thus the interval behavior comes from a continuous family of admissible code subspaces. Other interval examples arise from analogous constructions together with additional Pauli constraints that vanish on this family.

\subsubsection{Extremal singleton values.}
At the opposite extreme, some admissible error sets force all detectable signature coordinates to vanish, yielding $\Sigma_2(\boldsymbol{E})=\{0\}$. A canonical example is the even-parity projector
\begin{equation}
P_{\mathrm{even}}
:=\lvert 00\rangle\!\langle 00\rvert+\lvert 11\rangle\!\langle 11\rvert.
\label{eq:n2_Peven}
\end{equation}
Whenever a Pauli observable $E$ flips the parity subspace, i.e.\ maps $\mathrm{Ran}(P_{\mathrm{even}})$ into its orthogonal complement, one has $P_{\mathrm{even}} E P_{\mathrm{even}}=0$, so the corresponding signature coordinate vanishes. This is the basic mechanism behind the singleton $\{0\}$.

Conversely, singleton-$\{1\}$ behavior occurs when at least one Pauli in the tuple acts as $\pm I$ on the code space, fixing the scalar compression to unit magnitude. In the two-qubit setting the remaining coordinates are then forced to vanish. For example,
\begin{equation}
P_{Z_2=-1}
:=\lvert 01\rangle\!\langle 01\rvert+\lvert 11\rangle\!\langle 11\rvert
\label{eq:n2_PZ1minus}
\end{equation}
satisfies
\begin{equation}
\begin{aligned}
P_{Z_2=-1}\,Z_2\,P_{Z_2=-1} &= -\,P_{Z_2=-1},\\
P_{Z_2=-1}\,X_2\,P_{Z_2=-1} &= 0,\\
P_{Z_2=-1}\,Y_2\,P_{Z_2=-1} &= 0,
\label{eq:n2_Zanchored}
\end{aligned}
\end{equation}
so for the tuple $(X_2,Y_2,Z_2)$ one has $\lambda^*(P_{Z_2=-1})=1$.

\subsection{With swap symmetry}
We now impose the $n=2$ cyclic symmetry, which in this case coincides with the qubit swap. Let $T$ be the swap unitary
\begin{equation}
T\lvert ab\rangle=\lvert ba\rangle,
\qquad
C_2=\langle T\rangle .
\label{eq:n2_swap_def}
\end{equation}
Unlike the symmetry-compatible families studied later in the paper, this two-qubit subsection should be viewed primarily as a toy model for an externally imposed projector-level symmetry constraint on the detectable sets classified above. Its purpose is to illustrate how rigidity may already emerge when symmetry is imposed only on the projector, even if the error set \(\mathcal{E}\) is not itself stable under the same group action. The stronger compatibility requirement, in which the symmetry acts consistently on both the code and the detectable Pauli family, will be developed in later sections.

This toy setting also clarifies why the unconstrained classification no longer applies once symmetry is imposed. The classification developed above relies essentially on the absence of symmetry restrictions: admissible code spaces can then be moved freely by Clifford conjugation. By contrast, after fixing a symmetry group \(G\), admissible code spaces must be \(G\)-invariant, and admissibility is preserved only by unitaries that respect this symmetry constraint. Accordingly, the spectrum is no longer fully Clifford invariant. In general, one has only the covariance relation 
\begin{equation}
\Sigma^{(G)}_K(\boldsymbol{E})=\Sigma^{(U^\dagger GU)}_K(U^\dagger \boldsymbol{E} U),
\end{equation}
rather than the stronger invariance 
\begin{equation}
\Sigma^{(G)}_K(\boldsymbol{E})=\Sigma^{(G)}_K(U^\dagger \boldsymbol{E} U)
\end{equation}
for arbitrary \(U\). In particular, under swap symmetry the restricted two-qubit classification can no longer be organized by full Clifford equivalence, since unitaries outside the normalizer of the symmetry group need not preserve admissibility.

The Hilbert space decomposes into the $T$-invariant sectors defined by the projectors
\begin{equation}
\Pi_0=\frac12(I+T),
\qquad
\Pi_1=\frac12(I-T),
\label{eq:n2_swap_sectors}
\end{equation}
so that $\mathcal{H}=\mathcal{H}_0\oplus\mathcal{H}_1$ with $\mathcal{H}_\ell=\mathrm{Ran}(\Pi_\ell)$. Here $\dim\mathcal{H}_0=3$ is the symmetric sector and $\dim\mathcal{H}_1=1$ is the antisymmetric sector. An orthonormal basis is
\begin{equation}
\begin{aligned}
\mathcal{H}_0&=\text{span}\Bigl\{\lvert 00\rangle,\ \lvert 11\rangle,\
\lvert\psi^+\rangle:=\frac1{\sqrt2}(\lvert 01\rangle+\lvert 10\rangle)\Bigr\}, \\
\mathcal{H}_1&=\text{span}\Bigl\{\lvert\psi^-\rangle:=\frac1{\sqrt2}(\lvert 01\rangle-\lvert 10\rangle)\Bigr\}.
\label{eq:n2_swap_basis}
\end{aligned}   
\end{equation}

We distinguish two notions of swap symmetry: the stronger \emph{swap-basis} condition \(\operatorname{Ran}(P)\subseteq \mathcal H_0\), and the weaker \emph{swap-projector} condition \(TPT^\dagger=P\), which allows support in both sectors. For \(K=2\), since \(\dim \mathcal H_1=1\), any projector satisfying the swap-projector condition must decompose as $P=\Pi_0 P \Pi_0+\Pi_1 P \Pi_1$, so the possible sector ranks are restricted to
\begin{equation}
(K_0,K_1)\in\{(2,0),(1,1)\},
\qquad
K_\ell:=\operatorname{rank}(\Pi_\ell P\Pi_\ell).
\end{equation}
The case \((2,0)\) is exactly the swap-basis setting. In this setting, the rank-$2$ compression inside the three-dimensional sector $\mathcal H_0$ is governed by eigenvalue interlacing, which forces each scalar compression to equal the middle eigenvalue of the restricted operator. This rigidity reduces the attainable swap-basis spectra to $\{\emptyset,\{0\},\{1\}\}$. The remaining case \((1,1)\) is analyzed in the swap-projector subsection, where we show that it produces no additional attainable spectra beyond the swap-basis case.

\subsubsection{Swap-basis case}
In the swap-basis setting $\mathrm{Ran}(P)\subseteq\mathcal{H}_0$, the compression of a Pauli observable is a rank-$2$ Hermitian compression inside the three-dimensional invariant subspace $\mathcal{H}_0$. A fundamental constraint is eigenvalue interlacing: if $E$ has eigenvalues $\lambda_1\ge\lambda_2\ge\lambda_3$ on $\mathcal{H}_0$, then any scalar compression $PEP=\alpha P$ must satisfy $\alpha=\lambda_2.$

Many Pauli restrictions to $\mathcal{H}_0$ have spectrum $\{1,0,-1\}$, which therefore forces $\alpha=0$. For example, in the ordered basis $\{\lvert 00\rangle,\lvert\psi^+\rangle,\lvert 11\rangle\}$,
\begin{equation}
\bar X_2:=\Pi_0 X_2 \Pi_0
=\frac1{\sqrt2}
\begin{pmatrix}
0 & 1 & 0\\
1 & 0 & 1\\
0 & 1 & 0
\end{pmatrix}.
\label{eq:n2_X1_restricted}
\end{equation}
Since $\operatorname{spec}(\bar X_2)=\{1,0,-1\}$, any swap-basis code detecting $X_2$ must satisfy $P X_2 P=0$. By the same interlacing argument and unitary equivalences within $\cH_0$ induced by Clifford symmetries, the same conclusion applies to the restrictions of $X_1$, $Y_1$, $Y_2$, $Z_1$, $Z_2$ and their $Z$-dressed single-qubit variants that remain three-level on $\cH_0$. Therefore, for any swap-basis detectable error set $\mathcal{E}$ whose members all restrict to $\{1,0,-1\}$ on $\cH_0$, the swap-basis signature spectrum collapses to
\begin{equation}
\Sigma^{\mathrm{swap\text{-}basis}}_2(\boldsymbol{E})=\{0\}.
\label{eq:n2_swap_basis_zero}
\end{equation}

The only way to avoid the forced zero compression is when the middle eigenvalue equals $\pm1$. This occurs for the parity operator $Z_1Z_2$, whose restriction to $\mathcal{H}_0$ is
\begin{equation}
\overline{ZZ}:=\Pi_0(Z_1Z_2)\Pi_0
=
\begin{pmatrix}
1 & 0 & 0\\
0 & -1 & 0\\
0 & 0 & 1
\end{pmatrix}.
\label{eq:n2_ZZ_restricted}
\end{equation}
Since $\operatorname{spec}(\overline{ZZ})=\{1,1,-1\}$, the middle eigenvalue is $1$, and eigenvalue interlacing implies that any rank-$2$ swap-basis code detecting $Z_1Z_2$ must satisfy $P(Z_1Z_2)P=P$ with $\mathrm{Ran}(P)\subseteq\mathcal{H}_0$. Hence the scalar compression is fixed to magnitude $1$. Whenever an error set $\mathcal{E}$ contains $Z_1Z_2$ together with additional Pauli operators whose swap-basis compressions vanish, the signature norm is uniquely determined:
\begin{equation}
\Sigma^{\mathrm{swap\text{-}basis}}_2(\boldsymbol{E})=\{1\}.
\label{eq:n2_swap_basis_one}
\end{equation}
The even-parity projector $P_{\mathrm{even}}$ in~\eqref{eq:n2_Peven} realizes this situation, since it is swap-invariant and satisfies $P_{\mathrm{even}}(Z_1Z_2)P_{\mathrm{even}}=P_{\mathrm{even}}$ while the compressions of single-qubit $X/Y$ operators vanish.

Because a scalar compression $PEP=\alpha P$ fixes the code space to a specific two-dimensional eigenspace inside $\mathcal{H}_0$, two independent Pauli constraints will in general select incompatible subspaces. In such cases no swap-symmetric projector satisfies the same Knill--Laflamme conditions, $Q_2(\boldsymbol{E})\neq\emptyset$ but
\begin{equation}
Q_2(\boldsymbol{E})\cap\{P:[P,T]=0\}=\emptyset.
\end{equation}
Thus the system of scalar compression equations becomes overdetermined once the symmetry restriction is imposed. Consequently the continuous family of solutions present in the unconstrained case disappears, and the only possible swap-basis spectra at $n=2$ are $\emptyset, \{0\},\{1\}.$

\subsubsection{Swap-projector case}
We now show that, for $n=2$, enforcing swap symmetry at the projector level does not enlarge the attainable signature norms beyond the swap-basis case. Since every swap-basis projector is in particular swap-projector-symmetric, it remains only to prove the reverse inclusion for attainable norms. Let
\begin{equation}
\ket{\phi^\pm}:=\frac{1}{\sqrt2}(\ket{00}\pm\ket{11}),
\end{equation}
so that
\begin{equation}
\begin{aligned}
\mathcal H_0&=\mathrm{span}\{\ket{\phi^-},\ket{\phi^+},\ket{\psi^+}\},\\
\mathcal H_1&=\mathrm{span}\{\ket{\psi^-}\}.
\end{aligned}
\end{equation}
If $(K_0,K_1)=(2,0)$, then $\operatorname{Ran}(P)\subseteq\mathcal H_0$ and we are already in the swap-basis setting. It therefore suffices to consider $(K_0,K_1)=(1,1)$, in which case
\begin{equation}
P=\ket{\xi}\!\bra{\xi}+\ket{\psi^-}\!\bra{\psi^-}
\end{equation}
for some unit vector $\ket{\xi}\in\mathcal H_0$. Define the rank-$2$ swap-basis projector
\begin{equation}
\widetilde{P}:=\Pi_0-\ket{\xi}\!\bra{\xi},
\qquad
\operatorname{Ran}(\widetilde{P})=\mathcal H_0\cap\ket{\xi}^{\perp}.
\end{equation}
We claim that for every Pauli observable $E$ satisfying $PEP=\alpha_E P$, the projector $\widetilde{P}$ also satisfies a scalar compression equation
\begin{equation}
\widetilde{P}E\widetilde{P}=\beta_E \widetilde{P}
\end{equation}
with $|\beta_E|=|\alpha_E|$. This proves that every signature norm attained by a swap-projector-symmetric code is also attained by a swap-basis code.

We verify the claim by treating two cases separately.

First, consider the three Pauli observables $X_1X_2$, $Y_1Y_2$, and $Z_1Z_2$. In the Bell basis
\begin{equation}
\{\ket{\phi^+},\ket{\phi^-},\ket{\psi^+},\ket{\psi^-}\},
\end{equation}
they are diagonal:
\begin{equation}
\begin{aligned}
X_1X_2&=\mathrm{diag}(1,-1,1,-1),\\
Y_1Y_2&=\mathrm{diag}(-1,1,1,-1),\\
Z_1Z_2&=\mathrm{diag}(1,1,-1,-1). 
\end{aligned}
\end{equation}
Hence each of them preserves $\mathcal H_1$, acts on $\ket{\psi^-}$ with eigenvalue $-1$, and has a one-dimensional $(-1)$-eigenspace in $\mathcal H_0$, namely
\begin{equation}
\begin{aligned}
\mathcal{V}_{-1}(X_1X_2|_{\mathcal H_0})&=\mathrm{span}\{\ket{\phi^-}\},\\
\mathcal{V}_{-1}(Y_1Y_2|_{\mathcal H_0})&=\mathrm{span}\{\ket{\phi^+}\},\\
\mathcal{V}_{-1}(Z_1Z_2|_{\mathcal H_0})&=\mathrm{span}\{\ket{\psi^+}\}.    
\end{aligned}
\end{equation}
Therefore, if $PEP=\alpha_E P$ for one of these three operators, then necessarily $\alpha_E=-1$ and $\ket{\xi}$ is the corresponding $(-1)$-eigenvector in $\mathcal H_0$. It follows that $\widetilde{P}$ is the orthogonal complement of this $(-1)$-eigenvector inside $\mathcal H_0$, hence the $+1$-eigenspace of $E|_{\mathcal H_0}$. Thus $\widetilde{P}E\widetilde{P}=\widetilde{P}$, so in this case $\beta_E=1$ and $|\beta_E|=|\alpha_E|=1$.

Second, consider any other nonidentity two-qubit Pauli observable $E$. A direct inspection in the Bell basis shows that such an $E$ maps $\mathcal H_1$ into $\mathcal H_0$; equivalently,
\begin{equation}
E\ket{\psi^-}\in \mathcal H_0,
\qquad
\braket{\psi^-|E|\psi^-}=0.
\end{equation}
Hence the scalar compression equation $PEP=\alpha_E P$ forces
\begin{equation}
\alpha_E=0,
\qquad
\braket{\xi|E|\psi^-}=0,
\qquad
\braket{\xi|E|\xi}=0.
\end{equation}
Let $\ket{v}:=E\ket{\psi^-}\in\mathcal H_0$. Then $\braket{\xi|v}=0$, so $\ket{v}\in\operatorname{Ran}(\widetilde{P})$. Since $E^2=I$, we have
\begin{equation}
E\ket{v}=\ket{\psi^-}\notin\operatorname{Ran}(\widetilde{P}),
\qquad
\widetilde{P}E\ket{v}=0.
\end{equation}
Let $\ket{\eta}$ be a unit vector spanning $\operatorname{Ran}(\widetilde{P})\ominus\mathbb C\ket{v}$. Because $E$ maps $\ket{v}\mapsto\ket{\psi^-}$, Hermiticity forces $E$ to preserve
\begin{equation}
\operatorname{span}\{\ket{\xi},\ket{\eta}\}=\ket{v}^{\perp}\cap\mathcal H_0.
\end{equation}
Together with $\braket{\xi|E|\xi}=0$, this gives $E\ket{\eta}\in\mathbb C\ket{\xi}$, hence $\widetilde{P}E\ket{\eta}=0$. Therefore
\begin{equation}
\widetilde{P}E\widetilde{P}=0,
\qquad
\beta_E=0=\alpha_E
\end{equation}
for every Pauli observable outside $\{X_1X_2,Y_1Y_2,Z_1Z_2\}$.

We have shown that for each Pauli observable $E$, whenever $PEP=\alpha_E P$ with $P$ of type $(1,1)$, the swap-basis projector $\widetilde{P}$ satisfies $\widetilde{P}E\widetilde{P}=\beta_E \widetilde{P}$ with $|\beta_E|=|\alpha_E|$. Since the same $\widetilde{P}$ is constructed from the given $P$ and works simultaneously for every $E\in\mathcal E$, it follows that every signature norm attained in the swap-projector setting is also attained in the swap-basis setting. The converse inclusion is immediate because every swap-basis projector satisfies the swap-projector condition. Hence
\begin{equation}
\Sigma^{\mathrm{swap\text{-}proj}}_2(\boldsymbol{E})
=
\Sigma^{\mathrm{swap\text{-}basis}}_2(\boldsymbol{E})
\in \bigl\{\emptyset,\{0\},\{1\}\bigr\}.
\label{eq:n2_swap_proj_equals_basis}
\end{equation}

Therefore, for $n=2$ the symmetry restriction collapses the continuous interval spectrum to the discrete set $\{\emptyset,\{0\},\{1\}\}$.

\section{Optimization-based Search Framework with Symmetry}
\label{sec:opt}

The exact two-qubit analysis of Sec.~\ref{sec:cases} makes the small-system geometry transparent, but the three-qubit random instances and the larger symmetry-compatible families studied later require numerical search \cite{CaoZhangWuGrasslZeng2022VarQEC}. In this section we describe the Stiefel-manifold framework used to explore the unrestricted spectrum $\Sigma_K(\boldsymbol{E})$ and its symmetry-restricted analogues $\Sigma_K^{(G)}(\boldsymbol{E})$. The framework combines a penalty-based search over rank-$K$ projectors with the symmetry-adapted decompositions developed in Sec.~\ref{sec:sym}. It is the method used throughout the numerical parts of Secs.~\ref{sec:n3_k2} and \ref{sec:state_vs_projector}.

\subsection{Stiefel parameterization and loss functions}
\label{subsec:stiefel-loss}

Fix a detectable Pauli family $\mathcal E$ and, when evaluating $\lambda^*$, an ordered tuple
\begin{equation}
\boldsymbol{E}=(E_1,\dots,E_m)\subseteq \mathcal E.  
\end{equation}
Let $D=2^n$. A rank-$K$ code projector is represented as
\begin{equation}
P(\Psi)=\Psi\Psi^\dagger,
\qquad
\Psi\in \mathrm{St}(D,K), 
\end{equation}
where $\Psi^\dagger\Psi=I_K$. We parameterize $\Psi$ by a free matrix $\theta\in\mathbb{C}^{D\times K}$ through the polar map
\begin{equation}
f:\mathbb{C}^{D\times K}\ni\theta\longmapsto \theta(\theta^\dagger\theta)^{-1/2}\in\mathrm{St}(D,K),
\label{eq:polar-map}
\end{equation}
and optimize over $\theta$ with gradient-based methods (Adam in our implementation) \cite{AbsilMahonySepulchre2008,KingmaBa2015Adam}.

For each $F\in\mathcal E$, define the code-space compression matrix
\begin{equation}
\begin{aligned}
M_F(\Psi)&:=\Psi^\dagger F\Psi\in\mathbb{C}^{K\times K},\\
\overline{\kappa}_F(\Psi)&:=\frac{1}{K}\operatorname{Tr} M_F(\Psi).
\end{aligned}
\end{equation}
The exact detection condition
\begin{equation}
PFP=\alpha_F P
\end{equation}
is equivalent to
\begin{equation}
M_F(\Psi)=\overline{\kappa}_F(\Psi)\,I_K.
\end{equation}
This motivates the Knill--Laflamme (KL) feasibility loss
\begin{equation}
\begin{aligned}
\mathcal{L}_{\mathrm{KL}}(\theta)
&=
\sum_{F\in\mathcal E}
\Biggl[
\sum_{i\neq j}\bigl|[M_F(\Psi)]_{ij}\bigr|^2 \\
&\quad
+\sum_i\bigl([M_F(\Psi)]_{ii}-\overline{\kappa}_F(\Psi)\bigr)^2
\Biggr],
\end{aligned}
\label{eq:LKL}
\end{equation}
where $\Psi = f(\theta)$. This loss vanishes if and only if $P(\Psi)$ detects every operator in $\mathcal E$.

For the ordered tuple $\boldsymbol E=(E_1,\dots,E_m)$ used in the signature norm, define
\begin{equation}
\begin{aligned}
\overline{\lambda}_\alpha(\Psi)&:=\frac{1}{K}\operatorname{Tr}(\Psi^\dagger E_\alpha\Psi),\\
\|\widetilde{\boldsymbol\lambda}(\Psi)\|_2^2&:=\sum_{\alpha=1}^m \overline{\lambda}_\alpha(\Psi)^2.
\end{aligned}
\end{equation}
On an exact detecting projector this quantity is precisely $\lambda^{*2}(P(\Psi))$. 
Given a penalty parameter $\mu>0$, we associate to the lower-endpoint search the objective
\begin{equation}
\mathcal{L}_{\underline{\lambda}}(\theta;\mu)
=\mu\,\mathcal{L}_{\mathrm{KL}}(\theta)
+\|\widetilde{\boldsymbol\lambda}(\Psi)\|_2^2.
\label{eq:Lmin}
\end{equation}
For the upper-endpoint search, we use the objective
\begin{equation}
\mathcal{L}_{\overline{\lambda}}(\theta;\mu)
=\mu\,\mathcal{L}_{\mathrm{KL}}(\theta)
-\|\widetilde{\boldsymbol\lambda}(\Psi)\|_2^2.
\label{eq:Lmax}
\end{equation}
For a prescribed target value $\lambda^*$, we further use the scanning objective
\begin{equation}
\mathcal{L}(\theta;\mu,\lambda^*)
=\mu\,\mathcal{L}_{\mathrm{KL}}(\theta)
+\bigl(\|\widetilde{\boldsymbol\lambda}(\Psi)\|_2^2-\lambda^{*2}\bigr)^2.
\label{eq:Lscan}
\end{equation}
For sufficiently large $\mu$, the optimizers of \eqref{eq:Lmin} and \eqref{eq:Lmax} approximate the spectral endpoints $\lambda_{\min}$ and $\lambda_{\max}$. 
A target grid for $\lambda^*$ in \eqref{eq:Lscan} is then used to probe the interior of $\Sigma_K(\boldsymbol E)$.

\subsection{Symmetry-adapted parameterizations}
\label{subsec:sym-param}

Whenever an exact symmetry reduction is available, we impose symmetry by construction rather than by penalty. This replaces the full ambient Stiefel manifold by smaller symmetry-adapted ones derived from the decompositions in Sec.~\ref{sec:sym}. We use three such reductions.

\paragraph{State-level reduction (embedding model).}
Let $B\in\mathbb{C}^{D\times d}$ be an isometry whose columns span the relevant invariant subspace under the group action, for example $\mathcal H_0$ in the cyclic case or $\mathcal H_{\mathrm{sym}}$ in the permutation-invariant case. Every state-level symmetric code frame has the form
\begin{equation}
\Psi=B\,\Phi,
\qquad
\Phi\in\mathrm{St}(d,K),
\label{eq:embedding}
\end{equation}
and all losses are evaluated on the reduced operators
\begin{equation}
\widetilde{E}_\alpha:=B^\dagger E_\alpha B\in\mathbb{C}^{d\times d}.
\end{equation}
This replaces the ambient dimension $D$ by $d$, which can be much smaller; for instance, $\dim\mathcal H_{\mathrm{sym}}=n+1$ whereas $D=2^n$.

\paragraph{Projector-level cyclic reduction.}
By Sec.~\ref{sec:sym}, a cyclic-invariant projector is block diagonal with respect to the character sectors
\begin{equation}
\mathcal H=\bigoplus_{\ell=0}^{n-1}\mathcal H_\ell.
\end{equation}
Fix a rank allocation $(K_0,\dots,K_{n-1})$ satisfying
\begin{equation}
\sum_{\ell=0}^{n-1}K_\ell=K,
\qquad
0\le K_\ell\le d_\ell:=\dim\mathcal H_\ell. 
\end{equation}
For each sector we optimize an independent block
\begin{equation}
\Phi_\ell\in\mathrm{St}(d_\ell,K_\ell).
\end{equation}
If $B_\ell\in\mathbb{C}^{D\times d_\ell}$ is a basis for $\mathcal H_\ell$, the assembled frame is
\begin{equation}
\Psi=\bigl[\,B_0\Phi_0\;\;B_1\Phi_1\;\;\cdots\;\;B_{n-1}\Phi_{n-1}\,\bigr].
\label{eq:cyclic-proj-frame}
\end{equation}
Because the sectors are mutually orthogonal, this $\Psi$ is automatically an isometry. The cyclic projector-level spectrum is obtained by taking the union over all admissible rank allocations.

\paragraph{Projector-level permutation reduction.}
For permutation symmetry we use the Schur--Weyl decomposition from Sec.~\ref{sec:sym}. Writing the Hilbert space as a direct sum of spin blocks, a permutation-invariant projector takes the form
\begin{equation}
P=\bigoplus_j \bigl(P_j\otimes I_{\mathcal M_j}\bigr),
\qquad
P_j=V_jV_j^\dagger,
\end{equation}
where \(V_j\in \mathrm{St}(2j+1,r_j)\), and the multiplicities $m_j:=\dim \mathcal M_j$ are constrained by
\begin{equation}
\sum_j r_j\,m_j=K.  
\end{equation}
If $B_j\in\mathbb{C}^{D\times(2j+1)m_j}$ is a basis for the $j$th block, arranged in multiplicity-major order, then the code frame is
\begin{equation}
\Psi=\bigl[\,\cdots\;\;B_j\,(V_j\otimes I_{\mathcal M_j})\;\;\cdots\,\bigr].
\label{eq:pi-proj-frame}
\end{equation}
The optimization variables therefore live on much smaller Stiefel manifolds, and the full permutation projector-level spectrum is obtained as the union over all admissible rank profiles $\{r_j\}$.

\paragraph{Soft symmetry penalty.}
When an exact symmetry reduction is unavailable or inconvenient, we keep the full $D$-dimensional search space and add a symmetry penalty. If $\mathcal G$ is a generating set for the symmetry group $G$, define
\begin{equation}
\mathcal{L}_{\mathrm{sym}}(\Psi):=
\sum_{g\in\mathcal G}
\bigl\|P(\Psi)-U_gP(\Psi)U_g^\dagger\bigr\|_{\mathrm F}^2.
\label{eq:Lsym}
\end{equation}
The task objective is then replaced by
\begin{equation}
\mathcal L_{\mathrm{tot}}
=
\mathcal L_{\mathrm{task}}
+\mu\,\mathcal L_{\mathrm{KL}}
+\mu_{\mathrm{sym}}\,\mathcal L_{\mathrm{sym}}.
\end{equation}
This does not reduce the search dimension, but it still biases the optimizer toward the required symmetry class.

\subsection{Optimization workflow}
\label{subsec:opt-workflow}

All numerical experiments in this paper follow the same branchwise protocol. A \emph{branch} means either a single unrestricted/state-level search, or a fixed sector/rank allocation in a projector-level symmetry reduction. For each branch, we first estimate the spectral endpoints by optimizing \eqref{eq:Lmin} and \eqref{eq:Lmax} with multiple random restarts. We then optionally probe the interior of the spectrum by scanning a prescribed target grid for $\lambda^*$, implemented numerically as a grid in $\lambda^{*2}$, using \eqref{eq:Lscan}. Every accepted candidate is validated by direct evaluation of the KL residuals and, when relevant, the symmetry residuals. The resulting validated $\lambda^*$-values give an empirical reconstruction of $\Sigma_K(\boldsymbol E)$ or $\Sigma_K^{(G)}(\boldsymbol E)$.

\begin{algorithm}[t]
\caption{Symmetry-adapted reconstruction of $\Sigma_K(\boldsymbol{E})$}
\label{alg:symmetry-spectrum}
\DontPrintSemicolon
\KwIn{$(n,K,\mathcal E,\boldsymbol E)$, search mode, penalties $(\mu,\mu_{\mathrm{sym}})$, restart count $N_{\mathrm{init}}$, optional target grid $\Gamma$, validation tolerances $(\varepsilon_{\mathrm{KL}},\varepsilon_{\mathrm{sym}})$.}
\KwOut{Validated set $S\subseteq \mathbb R_{\ge 0}$ of attainable $\lambda^*$-values.}
Construct the search branches: a single branch in the unconstrained and state-level modes, and all admissible sector/rank allocations in the cyclic or permutation projector-level modes\;
$S\leftarrow \emptyset$\;
\ForEach{branch $b$}{
  Run $N_{\mathrm{init}}$ random restarts of $\mathcal L_{\underline{\lambda}}$ and $\mathcal L_{\overline{\lambda}}$ on $b$\;
  Validate the best endpoint candidates and add their $\lambda^*$-values to $S$\;
}
\If{$\Gamma\neq \emptyset$}{
  \ForEach{target $\lambda^*\in \Gamma$}{
    \ForEach{branch $b$}{
      Run $N_{\mathrm{init}}$ random restarts of $\mathcal L(\theta;\mu,\lambda^*)$ on $b$\;
      \If{a converged candidate satisfies the KL and symmetry tolerances}{
        Add its $\lambda^*$-value to $S$ and stop scanning this target\;
      }
    }
  }
}
Deduplicate $S$ up to numerical tolerance and report its interval/singleton/empty form\;
\end{algorithm}

Our current scans are implemented in Python using \texttt{numqi} \cite{numqi2025}. In every mode we use multiple random restarts together with post-optimization validation, and the numerical table entries reported later should be understood as empirical outputs of this validated search protocol.

\section{The Three-qubit Case (\texorpdfstring{$n=3$, $K=2$}{n=3, K=2})}
\label{sec:n3_k2}

At $n=3$ and $K=2$, the signature geometry already separates into two qualitatively different regimes. For genuinely asymmetric random Pauli tuples, the unrestricted problem continues to exhibit the interval behavior of Observation~\ref{ob:interval}. By contrast, if one imposes a cyclic symmetry constraint on the code while keeping the Pauli tuple random and asymmetric, then the behavior can change drastically: the spectrum may shrink, collapse to a singleton, become empty, or even become disconnected. For symmetry-compatible Pauli families, however, cyclic symmetry acts by orbit reduction rather than by an external slicing condition, and the reduced problem remains highly structured. We discuss these two regimes separately.

\subsection{Random Pauli error sets}

Random three-qubit Pauli tuples provide a useful test of Observation~\ref{ob:interval} beyond the symmetry-compatible families discussed elsewhere in the paper. In this subsection, each tuple $\boldsymbol{E}=(E_1,\dots,E_m)$ is chosen directly from the non-identity Hermitian elements of $\mathcal P_3$, and $\Sigma_2(\boldsymbol{E})$ is explored numerically by the Stiefel-manifold procedure of Sec.~\ref{sec:opt}.

\subsubsection{Unrestricted random tuples}

We first impose no symmetry constraint on the code projector. Thus $P=\Psi\Psi^\dagger$ ranges over all rank-$2$ projectors with $\Psi\in St(8,2)$, and detectability means $P E_a P=\alpha_a P$ for every listed Pauli operator. We sampled $1002$ random tuples: $334$ with $m=6$, $334$ with $m=7$, and $334$ with $m=8$.

The resulting picture is strikingly uniform. In every case the lower endpoint was numerically $0$, while the upper endpoint lay extremely close to $1$, $\sqrt2$, or $\sqrt3$; equivalently, $\lambda_{\max}^{*2}$ was very close to $1$, $2$, or $3$. Concretely, for $m=6$ the upper endpoint was near $1$ in $50$ cases, near $\sqrt2$ in $258$ cases, and near $\sqrt3$ in $26$ cases; for $m=7$ the corresponding counts were $43$, $248$, and $43$; and for $m=8$ they were $45$, $249$, and $40$. Across all $1002$ tuples, $\lambda_{\max}$ was near $1$ in $138$ cases, near $\sqrt2$ in $755$ cases, and near $\sqrt3$ in $109$ cases.

Representative dense rescans further corroborate this picture. Specifically, we use a uniform $21$-point grid in $\lambda^{*2}$ between the endpoint estimates $\lambda_{\min}^{*2}$ and $\lambda_{\max}^{*2}$, and probe each target via the scanning objective~\eqref{eq:Lscan}. Table~\ref{tab:app_n3_dense_unrestricted} lists one representative tuple for each of the three observed upper-endpoint values. In all three examples, every grid target was numerically reached. The unrestricted random data therefore further support Observation~\ref{ob:interval}: for the tested random three-qubit tuples, the scalar signature spectrum continues to behave like a single interval.

\begin{table}[ht]
    \centering
    \caption{Representative dense rescans for unrestricted random three-qubit tuples. The second column reports the numerically observed $\lambda^{*}$-spectrum.}
    \label{tab:app_n3_dense_unrestricted}
    \footnotesize
    \setlength{\tabcolsep}{3pt}
    \renewcommand{\arraystretch}{1.3}
    \begin{tabular}{|c|c|}
        \hline
        \textbf{Error set} & \textbf{$\lambda^{*}$-spectrum} \\
        \hline
        $\{YXY, ZXZ, IZX, ZYZ, YIZ, XXZ\}$ & $[0,1]$ \\
        \hline
        $\{ZZY, IIZ, IXZ, IIX, ZZX, IYY, XIX\}$ & $[0,\sqrt2]$ \\
        \hline
        $\{ZXX, IXY, XIX, YIY, YIZ, ZIY, YZI, IIZ\}$ & $[0,\sqrt3]$ \\
        \hline
    \end{tabular}
\end{table}

\subsubsection{Random tuples with an external cyclic-\texorpdfstring{$+1$}{+1} restriction}

We next impose \emph{state-level} cyclic symmetry on the code while keeping the Pauli tuple random and asymmetric. Let $T$ be the cyclic shift,
\begin{equation}
T\lvert b_1b_2b_3\rangle=\lvert b_3b_1b_2\rangle,
\end{equation}
and let $\mathcal H_0$ be the $+1$ eigenspace of $T$. For $n=3$ this is the four-dimensional space
\begin{equation}
\begin{aligned}
\mathcal H_0=\operatorname{span}\Bigl\{
&\lvert 000\rangle,\;
\lvert W\rangle:=\frac{\lvert001\rangle+\lvert010\rangle+\lvert100\rangle}{\sqrt3},\\
&\lvert\overline W\rangle:=\frac{\lvert011\rangle+\lvert101\rangle+\lvert110\rangle}{\sqrt3},\;
\lvert111\rangle
\Bigr\}.
\end{aligned}
\label{eq:n3_random_H0}
\end{equation}
Note that $\mathcal{H}_0$ coincides with the fully symmetric subspace $\mathcal{H}_{\mathrm{sym}}$ discussed in Sec.~\ref{sec:sym}. Since every permutation-invariant vector is in particular cyclic-invariant, we have $\mathcal{H}_{\mathrm{sym}} \subseteq \mathcal{H}_0$. For $n=3$, the Dicke states $\{|D_{3,w}\rangle\}_{w=0}^{3}$ span an $(n+1)=4$ dimensional space; since $\dim \mathcal{H}_0 = 4$, the equality follows:
\begin{equation}
\mathcal{H}_{\mathrm{sym}} = \mathcal{H}_0 \qquad (n=3).
\label{eq:n3_sym_equals_H0}
\end{equation}
This inclusion is strict ($\mathcal{H}_{\mathrm{sym}} \subsetneq \mathcal{H}_0$) for $n \ge 4$, as $\dim \mathcal{H}_0$ then exceeds $\dim \mathcal{H}_{\mathrm{sym}}$. Consequently, at $n=3$, state-level cyclic symmetry and full permutation invariance impose identical constraints on the code. We maintain the cyclic perspective throughout this section, as it generalizes more naturally to the projector-level analysis and the associated character-sector decomposition.

Restricting to cyclic-$+1$ codewords means that the projector has the form
\begin{equation}
P=B_0\Phi\Phi^\dagger B_0^\dagger,
\qquad \Phi\in St(4,2),
\end{equation}
where the columns of $B_0$ are the basis vectors in \eqref{eq:n3_random_H0}. The conceptual point is that this is not the symmetry-compatible situation of Sec.~\ref{sec:sym}, because a random tuple $\boldsymbol{E}$ is typically not stable under cyclic relabeling. The cyclic condition is therefore external to the error model rather than induced by it.

We tested $2000$ such restricted random tuples: $1000$ with $m=5$ and $1000$ with $m=6$. In this externally constrained problem the universal interval pattern disappears. Among the tested tuples we observed four behaviors:
\begin{enumerate}[
    label=\textbf{Case~\arabic*},
    leftmargin=5.0em,
    labelsep=0.8em,
    itemsep=0.45em
]
    \item Interval spectra;
    \item Singleton spectra;
    \item Empty spectra;
    \item Robust disconnected spectra, behaving like $\{0,1\}$.
\end{enumerate}

The first three cases are precisely the kinds of behavior already anticipated in the symmetry-compatible setting summarized in Fig.~\ref{fig:sigma-schematic}: a symmetry restriction may make the spectrum empty, collapse it to a singleton, or leave a genuine interval.  \textbf{Case~4}, however, is qualitatively different.  It reflects the fact that the cyclic-$+1$ condition is being imposed externally on a Pauli tuple that is not itself cyclic-stable, so the restriction can slice the feasible set into disconnected pieces.  In the next subsection we exhibit a concrete random tuple of this type and prove analytically that its cyclic-$+1$ restricted spectrum is exactly $\{0,1\}$.

The numerical breakdown across the $1000$ tuples for each $m$ is as follows.
For $m=5$: $253$ tuples had empty spectra, $511$ had singleton spectra, $132$ had genuine interval spectra, and $104$ had disconnected spectra.
For $m=6$: $372$ tuples had empty spectra, $486$ had singleton spectra, $78$ had genuine interval spectra, and $64$ had disconnected spectra.
Table~\ref{tab:app_n3_cyclic_m56} lists one representative tuple for each behavior and each value of $m$.

\begin{table}[ht]
    \centering
    \caption{Representative cyclic-$+1$ restricted random tuples illustrating each of the four observed spectral behaviors, for both $m=5$ and $m=6$. The second column reports the numerically observed $\lambda^{*}$-spectrum.}
    \label{tab:app_n3_cyclic_m56}
    \footnotesize
    \setlength{\tabcolsep}{3pt}
    \renewcommand{\arraystretch}{1.3}
    \begin{tabular}{|c|c|c|}
        \hline
        \textbf{$m$} & \textbf{Error set} & \textbf{$\lambda^{*}$-spectrum} \\
        \hline
        $5$ & $\{XXX,IZX,IXI,YYI,YIZ\}$ & $\emptyset$ \\
        \hline
        $6$ & $\{XXI,YXZ,ZZZ,YIX,IZI,IYZ\}$ & $\emptyset$ \\
        \hline
        $5$ & $\{XIY,ZXY,YYZ,XXY,ZXI\}$ & $\{0\}$ \\
        \hline
        $6$ & $\{IIX,YIX,IYX,YYZ,XIX,XXY\}$ & $\{0\}$ \\
        \hline
        $5$ & $\{IZX,XIY,XIZ,IXZ,IYI\}$ & $[0,\sqrt{0.75}]$ \\
        \hline
        $6$ & $\{IZX,XIY,XIZ,IXZ,IYI,YXX\}$ & $[0,\sqrt{0.0022}]$ \\
        \hline
        $5$ & $\{YXX,XXI,YXZ,YIX,IZI\}$ & $\{0,1\}$ \\
        \hline
        $6$ & $\{YXX,YYI,ZII,IIZ,IZY,YIX\}$ & $\{0,1\}$ \\
        \hline
    \end{tabular}
\end{table}

\subsubsection{An exact disconnected cyclic-$+1$ spectrum}

A representative disconnected example in Table~\ref{tab:app_n3_cyclic_m56} is the tuple
\begin{equation}
\boldsymbol{E}_{\mathrm{disc}}=(YXX,XXI,YXZ,YIX,IZI)
\label{eq:n3_disc_tuple}
\end{equation}
which has an \emph{exactly} disconnected cyclic-$+1$ spectrum: the only attainable values are $0$ and $1$. Thus, already at $n=3$, the difference between symmetry-compatible and externally imposed symmetry constraints is visible at the level of the scalar signature spectrum.

We now prove that the random tuple in \eqref{eq:n3_disc_tuple} has an exactly disconnected cyclic-$+1$ spectrum.

\begin{proposition}
Let
\begin{equation}
\boldsymbol{E}_{\mathrm{disc}}=(YXX,XXI,YXZ,YIX,IZI)
\end{equation}
and restrict to rank-$2$ projectors $P$ supported on the cyclic-$+1$ sector $\mathcal H_0$. Then the attainable signature spectrum is
\begin{equation}
\begin{aligned}
\left\{
\lambda^*(P)
:
\begin{array}{l}
P\subseteq\mathcal H_0,P E P=\alpha_E P\\
\text{for all } E\in\boldsymbol{E}_{\mathrm{disc}}
\end{array}
\right\}
=\{0,1\}.
\end{aligned}
\end{equation}
In particular, the cyclic-$+1$ spectrum is nonempty and disconnected.
\end{proposition}

\begin{proof}
Write the ordered basis of $\mathcal H_0$ as
\begin{equation}
\begin{aligned}
\lvert000\rangle&\leftrightarrow\lvert00\rangle,\qquad
\lvert W\rangle\leftrightarrow\lvert01\rangle,\\
\lvert\overline W\rangle&\leftrightarrow\lvert10\rangle,\qquad
\lvert111\rangle\leftrightarrow\lvert11\rangle.
\end{aligned}
\end{equation}
so that the reduced problem becomes a $4$-dimensional one. Replacing each reduced operator by \(U^\dagger \widetilde{E}\, U\) with \(U=(SH)\otimes I\) leaves the Knill--Laflamme equations and the compression scalars unchanged. In the resulting basis, the five reduced observables take the effective two-qubit Pauli form
\begin{equation}
\begin{aligned}
\widetilde E_{YXX}'&=\frac23 ZX+\frac13 YY,\\
\widetilde E_{XXI}'&=\frac13 I-\frac13 XZ+\frac1{\sqrt3}YI,\\
\widetilde E_{YXZ}'&=\frac1{\sqrt3}ZZ,\\
\widetilde E_{YIX}'&=\frac1{\sqrt3}ZI,\\
\widetilde E_{IZI}'&=\frac23 XI+\frac13 IZ.
\end{aligned}
\label{eq:app_disc_effective_ops}
\end{equation}

Let $P=VV^\dagger$ with $V\in\C^{4\times2}$ an isometry, and let $r_j$ be the rows of $V$. Define $R_j:=r_j^\dagger r_j$ for $j=1,2,3,4$. Then each $R_j\succeq0$ has rank at most one and
\begin{equation}
R_1+R_2+R_3+R_4=I_2.
\label{eq:app_disc_row_sum}
\end{equation}
The scalar-compression equations for the two diagonal observables $ZI$ and $ZZ$ in \eqref{eq:app_disc_effective_ops} read
\begin{equation}
\begin{aligned}
R_1+R_2-R_3-R_4&=\beta I_2,\\
R_1-R_2-R_3+R_4&=\gamma I_2  
\end{aligned}
\label{eq:app_disc_diag_eqs}
\end{equation}
for some real $\beta,\gamma$. Adding and subtracting gives
\begin{equation}
R_1-R_3=\frac{\beta+\gamma}{2}I_2,
\qquad
R_2-R_4=\frac{\beta-\gamma}{2}I_2.
\end{equation}
But the difference of two positive semidefinite matrices of rank at most one cannot be a nonzero multiple of $I_2$. Hence $\beta=\gamma=0$, and therefore
\begin{equation}
R_1=R_3,
\qquad
R_2=R_4,
\qquad
R_1+R_2=\frac12 I_2.
\label{eq:app_disc_row_reduction}
\end{equation}
Since $R_1$ and $R_2$ are rank-one positive semidefinite matrices summing to $I_2/2$, they are orthogonal rank-one projectors of weight $1/2$. After right-multiplying $V$ by a $2\times2$ unitary acting only inside the code space, we may assume
\begin{equation}
V(\phi,\psi)=\frac1{\sqrt2}
\begin{pmatrix}
1&0\\
0&1\\
e^{i\phi}&0\\
0&e^{i\psi}
\end{pmatrix}
\end{equation}
for some phases $\phi,\psi\in\R$. Equivalently, the code basis vectors are
\begin{equation}
v_1=\frac{\lvert00\rangle+e^{i\phi}\lvert10\rangle}{\sqrt2},
\qquad
v_2=\frac{\lvert01\rangle+e^{i\psi}\lvert11\rangle}{\sqrt2}.
\label{eq:app_disc_code_basis}
\end{equation}

We now impose the remaining three scalar-compression conditions.

For $\widetilde E_{IZI}'$, the off-diagonal code matrix element vanishes automatically, while equality of the two diagonal entries gives
\begin{equation}
\begin{aligned}
\langle v_1|\widetilde E_{IZI}'|v_1\rangle&=\frac{1+2\cos\phi}{3},\\
\langle v_2|\widetilde E_{IZI}'|v_2\rangle&=\frac{-1+2\cos\psi}{3}.
\end{aligned}
\end{equation}
and therefore
\begin{equation}
\cos\psi-\cos\phi=1.
\label{eq:app_disc_cos_eq}
\end{equation}

For $\widetilde E_{YXX}'$, the diagonal entries vanish automatically, and the off-diagonal one is
\begin{equation}
\langle v_1|\widetilde E_{YXX}'|v_2\rangle
=\frac{2e^{i\phi}-2e^{i\psi}-e^{i(\phi+\psi)}+1}{6e^{i\phi}}.
\end{equation}
Hence exact detectability forces
\begin{equation}
2e^{i\phi}-2e^{i\psi}-e^{i(\phi+\psi)}+1=0.
\label{eq:app_disc_exp_eq}
\end{equation}
Introduce
\begin{equation}
u:=\frac{\phi+\psi}{2},
\qquad
v:=\frac{\psi-\phi}{2}.
\end{equation}
Then \eqref{eq:app_disc_exp_eq} is equivalent to
\begin{equation}
\sin u=-2\sin v,
\label{eq:app_disc_sin1}
\end{equation}
while \eqref{eq:app_disc_cos_eq} becomes
\begin{equation}
-2\sin u\,\sin v=1.
\label{eq:app_disc_sin2}
\end{equation}
Substituting \eqref{eq:app_disc_sin1} into \eqref{eq:app_disc_sin2} gives
\begin{equation}
4\sin^2 v=1,
\end{equation}
so $\sin v=\pm1/2$ and $\sin u=\mp1$. Modulo $2\pi$, this leaves exactly two phase pairs:
\begin{equation}
(\phi,\psi)=\left(\frac{2\pi}{3},\frac{\pi}{3}\right)
\qquad\text{or}\qquad
(\phi,\psi)=\left(\frac{4\pi}{3},\frac{5\pi}{3}\right).
\label{eq:app_disc_two_phase_pairs}
\end{equation}

It remains to evaluate the compression scalar for $\widetilde E_{XXI}'$. One finds
\begin{equation}
\langle v_1|\widetilde E_{XXI}'|v_1\rangle
=\frac13-\frac13\cos\phi+\frac1{\sqrt3}\sin\phi,
\end{equation}
\begin{equation}
\langle v_2|\widetilde E_{XXI}'|v_2\rangle
=\frac13+\frac13\cos\psi+\frac1{\sqrt3}\sin\psi.
\end{equation}
For the first phase pair in \eqref{eq:app_disc_two_phase_pairs}, both expressions equal $1$; for the second, both equal $0$. Thus the only feasible compression vectors are
\begin{equation}
(\alpha_{YXX},\alpha_{XXI},\alpha_{YXZ},\alpha_{YIX},\alpha_{IZI})=(0,1,0,0,0)
\end{equation}
and
\begin{equation}
(\alpha_{YXX},\alpha_{XXI},\alpha_{YXZ},\alpha_{YIX},\alpha_{IZI})=(0,0,0,0,0).
\end{equation}
Consequently the only attainable signature norms are
\begin{equation}
\lambda^*=1\qquad\text{and}\qquad \lambda^*=0.
\end{equation}
This proves that the cyclic-$+1$ spectrum is exactly $\{0,1\}$.
\end{proof}

The two corresponding cyclic-$+1$ code spaces can be written explicitly in the original three-qubit basis as
\begin{equation}
\mathcal C_1=\operatorname{span}\left\{
\begin{aligned}
&\frac{\lvert000\rangle+\lvert011\rangle+\lvert101\rangle+\lvert110\rangle}{2},\\
&\frac{\lvert001\rangle+\lvert010\rangle+\lvert100\rangle+\lvert111\rangle}{2}
\end{aligned}
\right\},
\end{equation}
for which $\lambda^*=1$, and
\begin{equation}
\mathcal C_0=\operatorname{span}\left\{
\begin{aligned}
&\frac{\lvert000\rangle-\lvert011\rangle-\lvert101\rangle-\lvert110\rangle}{2},\\
&\frac{\lvert001\rangle+\lvert010\rangle+\lvert100\rangle-\lvert111\rangle}{2}
\end{aligned}
\right\},
\end{equation}
for which $\lambda^*=0$.

\subsection{Pauli error sets with symmetry}

We now return to four representative three-qubit Pauli families that are stable under cyclic relabeling. For these families, symmetry acts by orbit reduction and the reduced problem can be analyzed explicitly. Let
\begin{equation}
\begin{aligned}
\mathcal E_1&:=\{X_i,Z_i\}_{i=1}^3,\\
\mathcal E_2&:=\{X_i,Y_i\}_{i=1}^3\cup\{Z_iZ_j\}_{1\le i<j\le 3},\\
\mathcal E_3&:=\mathcal E_2\cup\{Z_1Z_2Z_3\},\\
\mathcal E_4&:=\{X_i,Z_i\}_{i=1}^3\cup\{X_iX_j,Z_iZ_j\}_{i<j}\cup\{X_iZ_j\}_{i\neq j}.
\end{aligned}
\end{equation}
For each family, let $\boldsymbol{E}_i$ be any fixed ordering of the elements of $\mathcal E_i$. Since $\lambda^*$ is invariant under permutations of the tuple, the resulting spectra do not depend on the chosen ordering. The unconstrained and cyclic spectra are summarized in Table~\ref{tab:n3_cyclic_examples}. For these four families, the projector-level cyclic spectrum agrees with the state-level cyclic spectrum supported on $\mathcal H_0$, so a single cyclic column suffices.

\begin{table}[ht]
    \centering
    \caption{Representative cyclic-stable three-qubit Pauli families. For these four families, state-level and projector-level cyclic symmetry give the same signature spectrum.}
    \label{tab:n3_cyclic_examples}
    \footnotesize
    \renewcommand{\arraystretch}{1.3}
    \begin{tabular}{|c|c|c|}
        \hline
        \textbf{Error set} & \textbf{Unconstrained} $\Sigma_2(\boldsymbol{E}_i)$ & \textbf{Cyclic} $\Sigma_2^{(C_3)}(\boldsymbol{E}_i)$ \\
        \hline
        $\mathcal E_1$ & $[0,1]$ & $[0,1/\sqrt3]$ \\
        \hline
        $\mathcal E_2$ & $[0,\sqrt3]$ & $\{\sqrt3\}$ \\
        \hline
        $\mathcal E_3$ & $[0,\sqrt2]$ & $\emptyset$ \\
        \hline
        $\mathcal E_4$ & $\{0\}$ & $\{0\}$ \\
        \hline
    \end{tabular}
\end{table}

\subsubsection{Code projector with no symmetry}

Without imposing any symmetry on the code projector, the four families above already realize the basic interval and rigidity patterns that recur throughout the paper. In each case, the endpoint values are realized by explicit rank-$2$ stabilizer codes detecting the listed Pauli errors. For a stabilizer
\begin{equation}
\mathsf S=\langle g_1,g_2\rangle
\end{equation}
with commuting Pauli generators $g_1,g_2$, we write
\begin{equation}
P_{\mathsf S}:=\frac14(I+g_1)(I+g_2)
\end{equation}
for the associated stabilizer projector.

\paragraph{$\mathcal E_1=\{X_i,Z_i\}_{i=1}^3$.}
The unrestricted spectrum is
\begin{equation}
\Sigma_2(\boldsymbol{E}_1)=[0,1].
\end{equation}
A convenient choice of endpoint stabilizers is
\begin{equation}
\begin{aligned}
\mathsf S_1^{\min}&=\langle Y_1,\;Y_2Y_3\rangle,\\
\mathsf S_1^{\max}&=\langle X_1,\;Y_2Y_3\rangle.
\end{aligned}
\end{equation}
Thus $P_{\mathsf S_1^{\min}}$ realizes $\lambda^*=0$, while $P_{\mathsf S_1^{\max}}$ realizes $\lambda^*=1$.

An exact interpolating family is
\begin{equation}
\begin{aligned}
P_1(\theta)
&=\bigl(\lvert\phi(\theta)\rangle\!\langle\phi(\theta)\rvert\bigr)_1\otimes \frac{I+Y_2Y_3}{2},\\
\lvert\phi(\theta)\rangle
&=\frac{e^{i\theta/2}\lvert0\rangle+i e^{-i\theta/2}\lvert1\rangle}{\sqrt2},
\end{aligned}
\end{equation}
where $\theta\in[0,\pi/2]$.
A direct computation gives
\begin{equation}
\begin{aligned}
P_1(\theta)X_1P_1(\theta)&=\sin\theta\,P_1(\theta),\\
P_1(\theta)EP_1(\theta)&=0,
\qquad \forall\,E\in\mathcal E_1\setminus\{X_1\}.
\end{aligned}
\end{equation}
Hence
\begin{equation}
\lambda^*(P_1(\theta))=\sin\theta,
\end{equation}
so the full interval $[0,1]$ is attained, with
\begin{equation}
P_1(0)=P_{\mathsf S_1^{\min}},\qquad P_1(\pi/2)=P_{\mathsf S_1^{\max}}.
\end{equation}

\paragraph{$\mathcal E_2=\{X_i,Y_i\}_{i=1}^3\cup\{Z_iZ_j\}_{1\le i<j\le 3}$.}
Here the unrestricted spectrum expands to
\begin{equation}
\Sigma_2(\boldsymbol{E}_2)=[0,\sqrt3].
\end{equation}
A convenient choice of endpoint stabilizers is
\begin{equation}
\begin{aligned}
\mathsf S_2^{\min}&=\langle X_1Z_2Z_3,\;Z_1X_2Z_3\rangle,\\
\mathsf S_2^{\max}&=\langle Z_1Z_2,\;Z_2Z_3\rangle.
\end{aligned}
\end{equation}
Thus $P_{\mathsf S_2^{\min}}$ realizes $\lambda^*=0$, while
\begin{equation}
P_{\mathsf S_2^{\max}}
=
\lvert000\rangle\!\langle000\rvert+\lvert111\rangle\!\langle111\rvert
\end{equation}
realizes $\lambda^*=\sqrt3$.

A convenient exact two-parameter family is obtained from the orthonormal codewords
\begin{widetext}
\begin{equation}
\begin{aligned}
\lvert0_L(u,v)\rangle
&= \frac{\sqrt{(1+u)(1+v)}}{2}\lvert000\rangle +\frac{\sqrt{(1+u)(1-v)}}{2}\lvert100\rangle +\frac{\sqrt{(1-u)(1+v)}}{2}\lvert010\rangle -\frac{\sqrt{(1-u)(1-v)}}{2}\lvert110\rangle,\\
\lvert1_L(u,v)\rangle
&= \frac{\sqrt{(1+u)(1+v)}}{2}\lvert111\rangle +\frac{\sqrt{(1+u)(1-v)}}{2}\lvert011\rangle +\frac{\sqrt{(1-u)(1+v)}}{2}\lvert101\rangle -\frac{\sqrt{(1-u)(1-v)}}{2}\lvert001\rangle,
\end{aligned}
\end{equation}
\end{widetext}
with $u,v\in[0,1]$, and the projector
\begin{equation}
P_2(u,v)=\lvert0_L(u,v)\rangle\!\langle0_L(u,v)\rvert+\lvert1_L(u,v)\rangle\!\langle1_L(u,v)\rvert.
\end{equation}
For this family, the only nonzero compression coefficients are
\begin{equation}
\begin{aligned}
P_2(u,v)X_1P_2(u,v)&=u\sqrt{1-v^2}\,P_2(u,v),\\
P_2(u,v)X_2P_2(u,v)&=v\sqrt{1-u^2}\,P_2(u,v),
\end{aligned}
\end{equation}
and
\begin{equation}
\begin{aligned}
P_2(u,v)Z_1Z_2P_2(u,v)&=uv\,P_2(u,v),\\
P_2(u,v)Z_1Z_3P_2(u,v)&=v\,P_2(u,v),\\
P_2(u,v)Z_2Z_3P_2(u,v)&=u\,P_2(u,v),
\end{aligned}
\end{equation}
while
\begin{equation}
\begin{aligned}
P_2(u,v)X_3P_2(u,v)&=0,\\
P_2(u,v)Y_iP_2(u,v)&=0 \qquad (i=1,2,3).
\end{aligned}
\end{equation}
Hence
\begin{equation}
\begin{aligned}
\lambda^*(P_2(u,v))^2
&= u^2(1-v^2)+v^2(1-u^2)+(uv)^2\\
&\quad+u^2+v^2\\
&= 2u^2+2v^2-u^2v^2.
\end{aligned}
\end{equation}
In particular, the diagonal branch $u=v=t$ connects the two endpoint stabilizers:
\begin{equation}
\lambda^*(P_2(t,t))=t\sqrt{4-t^2},
\qquad t\in[0,1],
\end{equation}
with
\begin{equation}
P_2(0,0)=P_{\mathsf S_2^{\min}},\qquad
P_2(1,1)=P_{\mathsf S_2^{\max}}.
\end{equation}

\paragraph{$\mathcal E_3=\mathcal E_2\cup\{Z_1Z_2Z_3\}$.}
Adding the three-body parity operator reduces the unrestricted spectrum to
\begin{equation}
\Sigma_2(\boldsymbol{E}_3)=[0,\sqrt2].
\end{equation}
The lower endpoint can again be taken to be
\begin{equation}
\mathsf S_3^{\min}=\mathsf S_2^{\min}=\langle X_1Z_2Z_3,\;Z_1X_2Z_3\rangle,
\end{equation}
while a convenient upper endpoint is
\begin{equation}
\mathsf S_3^{\max}=\langle X_1,\;Z_2Z_3\rangle.
\end{equation}
Thus $P_{\mathsf S_3^{\min}}$ realizes $\lambda^*=0$, and $P_{\mathsf S_3^{\max}}$ realizes $\lambda^*=\sqrt2$.

The family in part (b) detects the additional operator $Z_1Z_2Z_3$ exactly only on the branches $uv=0$, since in the basis
\(
\{\lvert0_L(u,v)\rangle,\lvert1_L(u,v)\rangle\}
\)
its compressed matrix is $\operatorname{diag}(uv,-uv)$. Choosing the branch $v=0$ and writing $u=t\in[0,1]$ gives the exact interpolation
\begin{equation}
\begin{aligned}
\lvert0_L(t)\rangle
&= \frac{\sqrt{1+t}}{2}\bigl(\lvert000\rangle+\lvert100\rangle\bigr)\\
&\quad+\frac{\sqrt{1-t}}{2}\bigl(\lvert010\rangle-\lvert110\rangle\bigr),\\
\lvert1_L(t)\rangle
&= \frac{\sqrt{1+t}}{2}\bigl(\lvert111\rangle+\lvert011\rangle\bigr)\\
&\quad+\frac{\sqrt{1-t}}{2}\bigl(\lvert101\rangle-\lvert001\rangle\bigr).
\end{aligned}
\end{equation}
with projector
\begin{equation}
P_3(t)=\lvert0_L(t)\rangle\!\langle0_L(t)\rvert+\lvert1_L(t)\rangle\!\langle1_L(t)\rvert.
\end{equation}
For this branch, the only nonzero compression coefficients are
\begin{equation}
\begin{aligned}
P_3(t)X_1P_3(t)&=t\,P_3(t),\\
P_3(t)Z_2Z_3P_3(t)&=t\,P_3(t),
\end{aligned}
\end{equation}
while
\begin{equation}
P_3(t)EP_3(t)=0,
\qquad
\forall\,E\in\mathcal E_3\setminus\{X_1,Z_2Z_3\}.
\end{equation}
Hence
\begin{equation}
\lambda^*(P_3(t))=\sqrt2\,t,
\qquad t\in[0,1],
\end{equation}
with
\begin{equation}
P_3(0)=P_{\mathsf S_3^{\min}},\qquad
P_3(1)=P_{\mathsf S_3^{\max}}.
\end{equation}

\paragraph{$\mathcal E_4=\{X_i,Z_i\}_{i=1}^3\cup\{X_iX_j,Z_iZ_j\}_{i<j}\cup\{X_iZ_j\}_{i\neq j}$.}
In this case detectability is already rigid in the unrestricted problem:
\begin{equation}
\Sigma_2(\boldsymbol{E}_4)=\{0\}.
\end{equation}
Both endpoints therefore coincide. A convenient stabilizer witness is
\begin{equation}
\mathsf S_4^{\min}=\mathsf S_4^{\max}=\langle Y_1Y_2,\;Y_2Y_3\rangle.
\end{equation}
Every operator in $\mathcal E_4$ anticommutes with at least one generator of $\mathsf S_4^{\min}$, so
\begin{equation}
P_{\mathsf S_4^{\min}}\,E\,P_{\mathsf S_4^{\min}}=0,
\qquad
\forall\,E\in\mathcal E_4.
\end{equation}
Accordingly, the interpolation degenerates to the constant path
\begin{equation}
P_4(t)=P_{\mathsf S_4^{\min}},
\qquad
 t\in[0,1],
\end{equation}
for which
\begin{equation}
\lambda^*(P_4(t))\equiv 0.
\end{equation}

\subsubsection{Code projector with cyclic symmetry}

We now impose cyclic symmetry on these same four families. Let $T$ be the one-step cyclic shift and write $\omega=e^{2\pi i/3}$. The Hilbert space decomposes as
\begin{equation}
\mathcal H=\mathcal H_0\oplus\mathcal H_{\omega}\oplus\mathcal H_{\omega^2},
\end{equation}
with $\dim\mathcal H_0=4$ and $\dim\mathcal H_{\omega}=\dim\mathcal H_{\omega^2}=2$. For the present four families, the nontrivial charge sectors do not contribute additional cyclic detecting projectors: once the tuple contains one-body operators $X_i$, $Y_i$, or $Z_i$, the corresponding neutral one-body restrictions on $\mathcal H_{\omega}$ and $\mathcal H_{\omega^2}$ are already nonscalar. Hence the cyclic projector problem effectively reduces to the $\mathcal H_0$ block.

Inside $\mathcal H_0$, one passes from individual Pauli operators to their neutral orbit averages. The relevant reduced blocks are
\begin{equation}
\begin{aligned}
\bar X&:=\frac{X_1+X_2+X_3}{3},\\
\bar Y&:=\frac{Y_1+Y_2+Y_3}{3},\\
\bar Z&:=\frac{Z_1+Z_2+Z_3}{3},
\end{aligned}
\end{equation}
\begin{equation}
\begin{aligned}
\overline{ZZ}&:=\frac{Z_1Z_2+Z_2Z_3+Z_3Z_1}{3},\\
\overline{XX}&:=\frac{X_1X_2+X_2X_3+X_3X_1}{3},
\end{aligned}
\end{equation}
\begin{equation}
\begin{aligned}
\overline{XZ}
&:=\frac{X_1Z_2+X_2Z_3+X_3Z_1}{3}\\
&=\frac{Z_1X_2+Z_2X_3+Z_3X_1}{3}
\end{aligned}
\end{equation}
on $\mathcal H_0$.Their explicit $4\times4$ matrices are listed below. 

For the symmetry-compatible three-qubit families from Sec.~\ref{sec:n3_k2}, it is useful to record the explicit neutral cyclic blocks on the $+1$ sector $\mathcal H_0$. In the ordered basis $(\lvert000\rangle,\lvert W\rangle,\lvert\overline W\rangle,\lvert111\rangle)$ one has
\begingroup
\setlength{\arraycolsep}{3pt}
\begin{equation}
\bar X=
\begin{pmatrix}
0&1/\sqrt3&0&0\\
1/\sqrt3&0&2/3&0\\
0&2/3&0&1/\sqrt3\\
0&0&1/\sqrt3&0
\end{pmatrix}.
\end{equation}
\begin{equation}
\bar Y=
\begin{pmatrix}
0&-i/\sqrt3&0&0\\
i/\sqrt3&0&-2i/3&0\\
0&2i/3&0&-i/\sqrt3\\
0&0&i/\sqrt3&0
\end{pmatrix}.
\end{equation}
\begin{equation}
\begin{aligned}
\bar Z&=\operatorname{diag}\!\left(1,\frac13,-\frac13,-1\right),\\
\overline{ZZ}&=\operatorname{diag}\!\left(1,-\frac13,-\frac13,1\right),\\
Z_1Z_2Z_3&=\operatorname{diag}(1,-1,1,-1).
\end{aligned}
\end{equation}
\begin{equation}
\overline{XX}=
\begin{pmatrix}
0&0&1/\sqrt3&0\\
0&2/3&0&1/\sqrt3\\
1/\sqrt3&0&2/3&0\\
0&1/\sqrt3&0&0
\end{pmatrix}.
\end{equation}
\begin{equation}
\overline{XZ}=
\begin{pmatrix}
0&1/\sqrt3&0&0\\
1/\sqrt3&0&0&0\\
0&0&0&-1/\sqrt3\\
0&0&-1/\sqrt3&0
\end{pmatrix}.
\end{equation}
\endgroup
These matrices satisfy the collective-spin identities quoted in \eqref{eq:n3_collective_first} and \eqref{eq:n3_collective_second}:
\begin{equation}
\bar X=\frac23J_x, \qquad \bar Y=\frac23J_y, \qquad \bar Z=\frac23J_z,
\label{eq:n3_collective_first}
\end{equation}
\begin{equation}
\begin{aligned}
\overline{ZZ}&=\frac23J_z^2-\frac12I,\\
\overline{XX}&=\frac23J_x^2-\frac12I,\\
\overline{XZ}&=\frac13\{J_x,J_z\}.
\end{aligned}
\label{eq:n3_collective_second}
\end{equation}
Thus, the cyclic reduction of the three-qubit problem is governed not by a generic Pauli tuple, but by a small collective-spin algebra acting on the spin-3/2 sector. This is precisely why the symmetry-compatible three-qubit examples retain a rigid orbit-reduced structure, in sharp contrast with the externally constrained random tuples from Sec.~\ref{sec:n3_k2}.

\paragraph{$\mathcal E_1$.}
For a cyclic rank-$2$ projector $P\subseteq\mathcal H_0$, the compression coefficients are site-independent:
\begin{equation}
\begin{aligned}
\lambda_{X_1}&=\lambda_{X_2}=\lambda_{X_3}=:\lambda_X,\\
\lambda_{Z_1}&=\lambda_{Z_2}=\lambda_{Z_3}=:\lambda_Z.
\end{aligned}
\end{equation}
For any real $u,v$ with $u^2+v^2=1$, the observable $u\bar X+v\bar Z=2J_{\hat n}/3$ has eigenvalues $\{-1,-1/3,1/3,1\}$. Eigenvalue interlacing for rank-$2$ scalar compressions therefore gives
\begin{equation}
|u\lambda_X+v\lambda_Z|\le \frac13,
\qquad \forall\,u^2+v^2=1,
\end{equation}
and hence
\begin{equation}
\lambda_X^2+\lambda_Z^2\le \frac19.
\end{equation}
Since $\lambda^*(P)^2=3\lambda_X^2+3\lambda_Z^2$, every cyclic code satisfies
\begin{equation}
\lambda^*(P)\le \frac1{\sqrt3}.
\end{equation}
The full interval $[0,1/\sqrt3]$ is realized, so
\begin{equation}
\Sigma_2^{(C_3)}(\boldsymbol{E}_1)=[0,\frac{1}{\sqrt3}].
\end{equation}

\paragraph{$\mathcal E_2$.}
The symmetric two-dimensional subspace
\begin{equation}
\operatorname{span}\{\lvert000\rangle,\lvert111\rangle\}
\end{equation}
is an immediate cyclic code. On it, $\overline{ZZ}=I$ while $\bar X$ and $\bar Y$ map the subspace into $\operatorname{span}\{\lvert W\rangle,\lvert\overline W\rangle\}$ and hence compress to $0$. Therefore
\begin{equation}
(\lambda_X,\lambda_Y,\lambda_{ZZ})=(0,0,1), \qquad \lambda^*(P)=\sqrt3.
\end{equation}
The reduced cyclic problem admits no intermediate values, and one finds
\begin{equation}
\Sigma_2^{(C_3)}(\boldsymbol{E}_2)=\{\sqrt3\}.
\end{equation}

\paragraph{$\mathcal E_3$.}
Adding $Z_1Z_2Z_3$ destroys the previous cyclic code, because on $\operatorname{span}\{\lvert000\rangle,\lvert111\rangle\}$ one has
\begin{equation}
Z_1Z_2Z_3=\operatorname{diag}(1,-1),
\end{equation}
which is not a scalar compression. The other natural two-dimensional block $\operatorname{span}\{\lvert W\rangle,\lvert\overline W\rangle\}$ already fails because $\bar X$ and $\bar Y$ are nonscalar there. The reduced cyclic equations are therefore infeasible, so
\begin{equation}
\Sigma_2^{(C_3)}(\boldsymbol{E}_3)=\emptyset.
\end{equation}

\paragraph{$\mathcal E_4$.}
The cyclicly reduced tuple is, up to affine rescaling, the family
\begin{equation}
J_x,\qquad J_z,\qquad J_x^2,\qquad \{J_x,J_z\},\qquad J_z^2.
\end{equation}
Simultaneous scalar compression forces vanishing first moments and isotropic second moments in the $x$-$z$ plane. Translating back to the Pauli orbit averages gives zero compression coefficients for all listed observables, hence
\begin{equation}
\Sigma_2^{(C_3)}(\boldsymbol{E}_4)=\{0\}.
\end{equation}

The contrast with the random cyclic-$+1$ experiments from the previous subsection is therefore structural rather than merely numerical. For the symmetry-compatible families $\mathcal E_1,\dots,\mathcal E_4$, cyclic symmetry reorganizes the Knill--Laflamme equations into neutral orbit averages and reduces the problem to a low-dimensional collective-spin block. For a random asymmetric tuple, by contrast, imposing $P\subseteq\mathcal H_0$ is an external constraint on the code and can destroy the interval behavior seen in the unrestricted problem.

\section{Larger Systems, Asymmetric Codes and Symmetry Constraints}
\label{sec:state_vs_projector}

Sections~\ref{sec:cases} and~\ref{sec:n3_k2} established the basic interval, singleton, and emptiness patterns for two- and three-qubit codes. We now extend this analysis to larger systems, where both the choice of detectable Pauli family and the level at which symmetry is imposed on the code give rise to a richer variety of behaviors. Table~\ref{tab:symmetry_gaps} summarizes the corresponding signature spectrum. In every row of Table~\ref{tab:symmetry_gaps}, the error family is invariant under qubit permutations, and hence also under cyclic shifts. Thus each pair $(\mathcal E,G)$ is symmetry-compatible in the sense of Sec.~\ref{sec:sym}. Moreover, in all symmetry-compatible cases studied here, the interval phenomenon persists: whenever the attainable signature spectrum is nonempty, it is a single closed interval.

The column $d$ in Table~\ref{tab:symmetry_gaps} distinguishes distance-$2$ codes, with error set $\{X_i,Y_i,Z_i\}$, from asymmetric families built from the hierarchical family $\mathcal E^{\mathrm{asym}}_{n,r}$ or the mixed family $\mathcal E^{\mathrm{mix}}_{n}$ introduced in Sec.~\ref{sec:preliminaries}. The label $r$ records the largest body order among the included $Z$-type strings: $r=2$ includes all $Z_iZ_j$, $r=3$ further includes all $Z_iZ_jZ_k$, and so on; $r=-$ indicates that the error set does not belong to this hierarchy.

Recall from Sec.~\ref{sec:sym} that \emph{state-level} symmetry requires every basis vector to be individually $G$-invariant, $U_g|v_i\rangle=|v_i\rangle$, whereas \emph{projector-level} symmetry requires only $U_g P U_g^\dagger=P$ for the code projector. The former implies the latter, but not conversely: a $G$-invariant projector may support a code space that transforms under a nontrivial representation of~$G$. For a fixed row of Table~\ref{tab:symmetry_gaps}, let $\mathcal{E}$ denote the corresponding detectable set from the definitions above, and let $\boldsymbol{E}$ be any fixed ordering of the elements of $\mathcal E$. We then let
$\mathcal P^{\mathrm{det}}_{K,\mathcal E}:=\{P\in\mathcal P_K : P \text{ is } \mathcal E\text{-detecting}\}$,
$\mathcal P^{\mathrm{basis}}_{K,G}:=\{P\in\mathcal P_K : P \text{ is state-level } G\text{-symmetric}\}$,
and
$\mathcal P^{\mathrm{proj}}_{K,G}:=\{P\in\mathcal P_K : P \text{ is projector-level } G\text{-symmetric}\}$.
Using these sets, we define the symmetry-restricted spectra
\begin{equation}
\Sigma^{\mathrm{basis}}_{K,G}(\boldsymbol{E})
:= \{\lambda^*(P): P\in\mathcal P^{\mathrm{det}}_{K,\mathcal E}\cap\mathcal P^{\mathrm{basis}}_{K,G}\},
\label{eq:sigma_models}
\end{equation}
\begin{equation}
\Sigma^{\mathrm{proj}}_{K,G}(\boldsymbol{E})
:= \{\lambda^*(P): P\in\mathcal P^{\mathrm{det}}_{K,\mathcal E}\cap\mathcal P^{\mathrm{proj}}_{K,G}\}.
\end{equation}
By construction, $\Sigma^{\mathrm{proj}}_{K,G}(\boldsymbol{E})$ coincides with the symmetry-restricted spectrum $\Sigma_K^{(G)}(\boldsymbol{E})$ defined in Sec.~\ref{sec:sym}. The more explicit superscript is adopted here to contrast with the state-level variant.
By definition,
\begin{equation}
\Sigma^{\mathrm{basis}}_{K,G}(\boldsymbol{E})
\subseteq
\Sigma^{\mathrm{proj}}_{K,G}(\boldsymbol{E})
\subseteq
\Sigma_K(\boldsymbol{E}).
\end{equation}
Hence the key question in this section is how the spectrum changes when the code-side symmetry is relaxed from basis level to projector level within a fixed symmetry-compatible problem. Three qualitative effects are observed: interval expansion, singleton collapse, and non-emptiness recovery. Table~\ref{tab:symmetry_gaps} summarizes these effects across parameters. In the spectrum columns, we write intervals as $[a,b]$, singletons as $\{c\}$, and $\emptyset$ for empty spectrum.

\begin{table*}[t]
    \centering
    \caption{Signature spectrum configurations for various codes under symmetry-compatible cyclic and permutation restrictions, comparing state-level and projector-level code-side symmetry. The spectrum values reported in the table are numerical estimates.}
    \label{tab:symmetry_gaps}
    \renewcommand{\arraystretch}{1.3} 
    \resizebox{\textwidth}{!}{%
    \begin{tabular}{|c|c|c|c|c|c|c|c|c|c|}
        \hline
        \textbf{r} & \textbf{n} & \textbf{K} & \textbf{d} & \textbf{Error set} & \textbf{Unconstrained} & \textbf{Cyclic basis} & \textbf{Cyclic projector} & \textbf{PI basis} & \textbf{PI projector} \\
        \hline
        - & 4 & 2 & 2 & $\{X_i, Y_i, Z_i\}$ & $\{0\}$ & $\{0\}$ & $\{0\}$ & $\{0\}$ & $\{0\}$ \\
        \hline
        - & 4 & 3 & 2 & $\{X_i, Y_i, Z_i\}$ & $\{0\}$ & $\{0\}$ & $\{0\}$ & $\emptyset$ & $\{0\}$ \\
        \hline
        - & 4 & 4 & 2 & $\{X_i, Y_i, Z_i\}$ & $\{0\}$ & $\emptyset$ & $\{0\}$ & $\emptyset$ & $\{0\}$ \\
        \hline
        - & 5 & 2 & 2 & $\{X_i, Y_i, Z_i\}$ & $[0, 1]$ & $[0, \sqrt{0.42}]$ & $[0, \sqrt{0.64}]$ & $[0, \sqrt{0.34}]$ & $[0, \sqrt{0.34}]$ \\
        \hline
        - & 5 & 3 & 2 & $\{X_i, Y_i, Z_i\}$ & $[0, 1]$ & $\{\sqrt{0.2}\}$ & $[0, \sqrt{0.57}]$ & $\emptyset$ & $\emptyset$ \\
        \hline
        - & 5 & 4 & 2 & $\{X_i, Y_i, Z_i\}$ & $[0, 1]$ & $\emptyset$ & $[0, \sqrt{0.2}]$ & $\emptyset$ & $\{0\}$ \\
        \hline
        - & 5 & 5 & 2 & $\{X_i, Y_i, Z_i\}$ & $\{0\}$ & $\emptyset$ & $\{0\}$ & $\emptyset$ & $\{0\}$ \\
        \hline
        - & 5 & 6 & 2 & $\{X_i, Y_i, Z_i\}$ & $\{0\}$ & $\emptyset$ & $\{0\}$ & $\emptyset$ & $\emptyset$ \\
        \hline
        - & 5 & 2 & - & $\{X_i, Y_i, Z_i, X_iX_j, Z_iZ_j, X_iZ_j, Z_iX_j\}$ & $[0, \sqrt{1.25}]$ & $[0, \sqrt{1.25}]$ & $[0, \sqrt{1.25}]$ & $\{\sqrt{1.25}\}$ & $\{\sqrt{1.25}\}$ \\
        \hline
        2 & 5 & 2 & - & $\{X_i, Y_i, Z_i, Z_iZ_j\}$ & $[0, \sqrt{3}]$ & $[0, \sqrt{2.5}]$ & $[0, \sqrt{2.5}]$ & $[\sqrt{0.31}, \sqrt{2.5}]$ & $[\sqrt{0.31}, \sqrt{2.5}]$ \\
        \hline
        3 & 5 & 2 & - & $\{X_i, Y_i, Z_i, Z_iZ_j, Z_iZ_jZ_k\}$ & $[0, 1]$ & $[0, \sqrt{0.5}]$ & $[0, \sqrt{0.5}]$ & $\emptyset$ & $\emptyset$ \\
        \hline
        4 & 5 & 2 & - & $\{X_i, Y_i, Z_i, Z_iZ_j, Z_iZ_jZ_k, Z_iZ_jZ_kZ_l\}$ & $\{0\}$ & $\{0\}$ & $\{0\}$ & $\emptyset$ & $\emptyset$ \\
        \hline
        2 & 5 & 3 & - & $\{X_i, Y_i, Z_i, Z_iZ_j\}$ & $\{0\}$ & $\emptyset$ & $\emptyset$ & $\emptyset$ & $\emptyset$ \\
        \hline
        2 & 5 & 4 & - & $\{X_i, Y_i, Z_i, Z_iZ_j\}$ & $\{0\}$ & $\emptyset$ & $\emptyset$ & $\emptyset$ & $\emptyset$ \\
        \hline
    \end{tabular}%
    } 
\end{table*}

\subsection{Cyclic symmetry}

Within a fixed cyclic-stable error model, relaxing to projector-level invariance allows the basis vectors to acquire complex phases under the cyclic shift, thereby enlarging the symmetry-restricted spectrum from $\Sigma^{\mathrm{basis}}_{K,C_n}(\boldsymbol{E})$ to $\Sigma^{\mathrm{proj}}_{K,C_n}(\boldsymbol{E})$.

\subsubsection{Interval expansion in $((5,2,2))$ codes.}
Let $n=5$ and $K=2$, and denote by $T$ the cyclic shift, a $5$-cycle acting on $(\mathbb{C}^2)^{\otimes 5}$. We write $\sum_{\mathrm{cyc}}$ for the sum over the five cyclic shifts of a computational basis string. When each logical basis vector is required to lie in the $+1$ eigenspace of $T$, so that $T|0_L\rangle=|0_L\rangle$ and $T|1_L\rangle=|1_L\rangle$, the attainable signature norm is analytically confined to the bounded interval
\begin{equation}
\lambda^*(P)
\in \left[0,\frac{\sqrt6-1}{\sqrt5}\right]
   \approx [0,0.6482].
\label{eq:522spectrum}
\end{equation}
and this entire interval is realized by a \emph{single} continuous cyclic-invariant family.

Specifically, for a real parameter $t$ in the range
\begin{equation}
0\le t\le \frac{\sqrt6-1}{5},
\label{eq:t_range_522}
\end{equation}
define the cyclic-invariant codewords
\begin{equation}
\begin{aligned}
|0_L(t)\rangle
&= a_0\,|00000\rangle + a_1\,|11111\rangle \\
&\quad + \frac{a_2}{\sqrt5}\sum_{\mathrm{cyc}}
\bigl(|00011\rangle - |00101\rangle\bigr),\\[2pt]
|1_L(t)\rangle
&= \frac{b_1}{\sqrt5}\sum_{\mathrm{cyc}} |00001\rangle \\
&\quad + \frac{b_2}{\sqrt5}\sum_{\mathrm{cyc}}
\bigl(|01011\rangle - |00111\rangle\bigr).
\end{aligned}
\label{eq:family_522_onepiece}
\end{equation}
with real non-negative coefficients satisfying
\begin{equation}
\begin{aligned}
a_0^2 &= \frac{3-2t-5t^2}{4(3+5t)},\qquad
b_1^2 = \frac{1+5t}{4},\\
a_1^2 &= \frac{1-2t-5t^2}{3+5t},\qquad
b_2^2 = \frac{3-5t}{8},\\
a_2^2 &= \frac{5(5t^2+6t+1)}{8(3+5t)}.
\end{aligned}
\label{eq:coeffs_522_onepiece}
\end{equation}
These choices ensure that $|0_L(t)\rangle$ and $|1_L(t)\rangle$ are normalized, mutually orthogonal, and each lies in the $+1$ eigenspace of $T$.

Let $P(t)=|0_L(t)\rangle\langle 0_L(t)|+|1_L(t)\rangle\langle 1_L(t)|$. A direct check of the Knill--Laflamme (KL) constraints for $\mathcal E=\{X_i,Y_i,Z_i\}_{i=1}^5$ shows that, for every site $i$,
\begin{equation}
\begin{aligned}
P(t)\,X_i\,P(t) &= P(t)\,Y_i\,P(t)=0,\\
P(t)\,Z_i\,P(t) &= t\,P(t).
\end{aligned}
\end{equation}
Hence the signature norm is
\begin{equation}
\lambda^*(P(t))=\left(\sum_{i=1}^{5}|t|^2\right)^{1/2}=\sqrt5\,t,
\end{equation}
so that as $t$ varies over \eqref{eq:t_range_522} the family \eqref{eq:family_522_onepiece} sweeps \emph{continuously} and \emph{exactly} the full cyclic-basis interval \eqref{eq:522spectrum}.

Finally, if one relaxes the requirement that each basis vector be $T$-invariant and instead requires only projector-level cyclic symmetry, 
\begin{equation}
TPT^\dagger=P, \qquad P=|0_L\rangle\langle 0_L|+|1_L\rangle\langle 1_L|,
\end{equation}
then $|0_L\rangle$ and $|1_L\rangle$ may transform with nontrivial phases under $T$ while the projector remains cyclic-symmetric. Writing $\omega=e^{2\pi i/5}$, a convenient cyclic-eigenvector construction is
\begin{equation}
|0_L\rangle\propto\sum_{k=0}^{4}\omega^{k}T^{k}|\psi_0\rangle,
\qquad
|1_L\rangle\propto\sum_{k=0}^{4}\omega^{-k}T^{k}|\psi_1\rangle. 
\end{equation}
With this relaxation, the feasible range extends beyond the analytic cyclic-basis bound, yielding the numerically observed interval
$\lambda^*(P)\in[0,\approx 0.8055]$. Equivalently,
\begin{equation}
\begin{aligned}
\Sigma^{\mathrm{basis}}_{2,C_5}(\boldsymbol{E})&=\left[0,\frac{\sqrt6-1}{\sqrt5}\right],\\
\Sigma^{\mathrm{proj}}_{2,C_5}(\boldsymbol{E})&\approx[0,0.8055].\\
\end{aligned}
\end{equation}


\subsubsection{Expansion from an isolated point in $((5,3,2))$.}
\label{sec:532_cyclic_gap}
For $\mathcal E=\{X_i,Y_i,Z_i\}_{i=1}^5$ and any fixed ordering $\boldsymbol E$, the cyclic-basis spectrum behaves differently from the interval expansion seen for $((5,2,2))$. Here, imposing \emph{state-level} cyclic symmetry forces each logical basis vector to lie in the $+1$ eigenspace of the cyclic shift $T$, and the attainable value of $\lambda^*(P)$ collapses to a single point. One explicit cyclic-invariant orthonormal basis is
\begin{equation}
\begin{aligned}
|0_L\rangle &= \sqrt{\frac35}\,|00000\rangle+\sqrt{\frac25}\,|11111\rangle,\\
|1_L\rangle &= \pm\frac{1}{\sqrt5}\sum_{\mathrm{cyc}}|00101\rangle,\\
|2_L\rangle &= \pm\frac{1}{\sqrt5}\sum_{\mathrm{cyc}}|00011\rangle.
\end{aligned}
\end{equation}
For each of these codewords, every qubit $i\in\{1,\dots,5\}$ is $1$ with probability $2/5$ and $0$ with probability $3/5$, so
\begin{equation}
\langle j_L|Z_i|j_L\rangle=\frac35-\frac25=\frac15,
\qquad j\in\{0,1,2\}.
\end{equation}
Moreover, $\langle j_L|X_i|j_L\rangle=\langle j_L|Y_i|j_L\rangle=0$, $\langle j_L|E|j'_L\rangle=0$ for $j\neq j'$ and all $E\in\{X_i,Y_i,Z_i\}_{i=1}^5$, since $X_i$ and $Y_i$ change Hamming weight and the three codewords occupy orthogonal cyclic weight sectors. Hence the only nonzero signature components are $\langle Z_i\rangle_{\rho_P}=1/5$, giving
\begin{equation}
\lambda^*(P)=\Bigl(\sum_{i=1}^5\Bigl|\frac15\Bigr|^2\Bigr)^{1/2}=\frac{1}{\sqrt5}.
\end{equation}
Thus
\begin{equation}
\Sigma^{\mathrm{basis}}_{3,C_5}(\boldsymbol E)=\left\{\frac{1}{\sqrt5}\right\}.
\end{equation}

To understand the projector-level expansion, it is useful to resolve the Hilbert space into $C_5$ character sectors. Let $\omega=e^{2\pi i/5}$ and let $T$ act by cyclic shift on computational basis states. Besides the two fixed strings $00000$ and $11111$, the remaining basis states split into six length-$5$ orbits generated by
$00001,\ 00011,\ 00101,\ 00111,\ 01011,\ 01111.$
For each nontrivial orbit representative $v$, define the Fourier-orbit states
\begin{equation}
|v;\ell\rangle=\frac{1}{\sqrt5}\sum_{r=0}^{4}\omega^{-\ell r}T^r|v\rangle,
\qquad \ell=0,1,2,3,4,
\end{equation}
so that $T|v;\ell\rangle=\omega^\ell |v;\ell\rangle$. This yields the decomposition
\begin{equation}
\begin{aligned}
(\mathbb C^2)^{\otimes 5}
&=
\bigoplus_{\ell=0}^{4}\mathcal H_\ell,\\
\bigl(\dim \mathcal H_\ell\bigr)_{\ell=0}^{4}
&=
(8,6,6,6,6).
\end{aligned}
\end{equation}
Hence every cyclic rank-$3$ projector is block diagonal,
\begin{equation}
\begin{aligned}
P &= \sum_{\ell=0}^{4} P_\ell,
\qquad
P_\ell := \Pi_\ell P \Pi_\ell,\\
r_\ell &:= \operatorname{rank}(P_\ell),
\qquad
\sum_{\ell=0}^{4} r_\ell = 3.
\end{aligned}
\end{equation}
where $\Pi_\ell$ is the projector onto $\mathcal H_\ell$. We write
\begin{equation}
\Sigma^{\mathrm{proj}}_{3,C_5}(\boldsymbol E;r_0,r_1,r_2,r_3,r_4) 
\end{equation}
for the projector-level spectrum inside a fixed support type $(r_0,r_1,r_2,r_3,r_4)$.

For $d=2$, the KL conditions only involve weight-$1$ Pauli operators. By cyclicity it is enough to impose the single-site compressions
\begin{equation}
\begin{aligned}
PX_1P&=\lambda_XP,\\
PY_1P&=\lambda_YP,\\
PZ_1P&=\lambda_ZP, 
\end{aligned}
\end{equation}
which then hold on every site. Consequently,
\begin{equation}
\lambda^{*2}(P)=5\bigl(|\lambda_X|^2+|\lambda_Y|^2+|\lambda_Z|^2\bigr).
\label{eq:532_irrep_lambda}
\end{equation}

The support type
\begin{equation}
(r_0,r_1,r_2,r_3,r_4)=(3,0,0,0,0)
\end{equation}
corresponds exactly to restricting $\mathrm{Ran}(P)\subseteq\mathcal H_0$, i.e. to the $+1$ eigenspace of $T$. This is the natural projector-level counterpart of the state-level cyclic condition. Numerically, both direct minimization and direct maximization on this branch collapse to the same value $\lambda_{\min}^{*2}\approx 0.2$, with $L_{\mathrm{KL}}\approx 6.67\times10^{-17}$ at both endpoints. Thus the entire projector search inside the $+1$ sector remains frozen at the same isolated point:
\begin{equation}
\Sigma^{\mathrm{proj}}_{3,C_5}(\boldsymbol E;3,0,0,0,0)
\approx
\left\{\frac{1}{\sqrt5}\right\}.
\end{equation}

The interval expansion appears only after support on nontrivial momentum sectors is allowed. Already the mixed support type
\begin{equation}
(r_0,r_1,r_2,r_3,r_4)=(1,1,0,0,1)  
\end{equation}
with one rank-$1$ block in each of $\mathcal H_0,\mathcal H_1,\mathcal H_4$ realizes the full projector-level interval found numerically. $\lambda_{\min}^{*2}$ and $\lambda_{\max}^{*2}$ are given by
\begin{equation}
\lambda_{\min}^{*2}\approx 3.7567\times10^{-12},
\qquad
\lambda_{\max}^{*2}\approx 0.5737,
\end{equation}
with corresponding residuals
\begin{equation}
L_{\mathrm{KL}}^{\min}\approx 3.98\times10^{-23},
\qquad
L_{\mathrm{KL}}^{\max}\approx 1.16\times10^{-16}.
\end{equation}
Equivalently,
\begin{equation}
\Sigma^{\mathrm{proj}}_{3,C_5}(\boldsymbol E;1,1,0,0,1)
\approx
[\,0,\;0.7575\,],
\label{eq:532_branch11001}
\end{equation}
or, in terms of $\lambda^{*2}$,
\begin{equation}
\lambda^{*2}\in[\,0,\;0.5737\,].
\end{equation}

To test whether \eqref{eq:532_branch11001} is genuinely an interval rather than just two isolated endpoint solutions, we also ran target optimizations for $\lambda^{*2}=t$ within the $(1,1,0,0,1)$ branch. Across the sampled range, the realized values track the targets to about $10^{-6}$ or better, while $L_{\mathrm{KL}}$ remains between about $10^{-19}$ and $10^{-21}$. This gives strong numerical evidence that the branch $(1,1,0,0,1)$ fills the entire interval in \eqref{eq:532_branch11001}.

Scanning all admissible sector-rank allocations gives the same global maximum, already attained on the $(1,1,0,0,1)$ branch, and a lower endpoint numerically indistinguishable from zero. Accordingly, the refined cyclic spectra are
\begin{equation}
\begin{aligned}
\Sigma^{\mathrm{basis}}_{3,C_5}(\boldsymbol E)
&=
\left\{\frac{1}{\sqrt5}\right\},\\
\Sigma^{\mathrm{proj}}_{3,C_5}(\boldsymbol E)
&\approx
[\,0,\;0.7575\,].
\end{aligned}
\end{equation}
No exact closed form is claimed for the upper endpoint. The structural point is that the interval beyond $1/\sqrt5$ is created entirely by support outside the trivial sector $\mathcal H_0$: even at projector level, the branch $\mathrm{Ran}(P)\subseteq\mathcal H_0$ remains a singleton at $\lambda^*=1/\sqrt5$, and the additional interval is produced by admitting nontrivial cyclic-momentum support, already through the conjugate pair $\mathcal H_1\oplus\mathcal H_4$.

\subsubsection{State--projector existence gap in $((4,4,2))$ and $((5,4,2))$.}
For $((4,4,2))$ and $((5,4,2))$, the $+1$ cyclic sector does not appear to contain a rank-$4$ projector satisfying the KL constraints, so no state-level cyclic basis code is obtained. Nevertheless, a cyclic-symmetric projector may still exist. For $((4,4,2))$, one such exact cyclic-symmetric projector at $\lambda^*(P)=0$ is
\begin{equation}
P=\frac{1}{4}(1+XXXX)(1+ZZZZ).
\end{equation}
For $((5,4,2))$, while no cyclic-invariant orthonormal basis is found, our computations indicate that cyclic-symmetric projectors exist throughout the interval
$\lambda^*(P)\in[0,1/\sqrt5]$. In spectrum language, this gives an existence gap with
\begin{equation}
\Sigma^{\mathrm{basis}}_{4,C_5}(\boldsymbol{E})=\emptyset,
\qquad
\Sigma^{\mathrm{proj}}_{4,C_5}(\boldsymbol{E})=\left[0,\frac{1}{\sqrt5}\right].
\end{equation}

\subsubsection{Asymmetric interval in $((5,2))$ codes with $\mathcal E^{\mathrm{asym}}_{5,2}$.}
\label{sec:52_asym_r2_cyclic}
At $n=5$ and $K=2$, we consider the hierarchical asymmetric family
\begin{equation}
\boldsymbol{E}:=\mathcal E^{\mathrm{asym}}_{5,2}
=\{X_i,Y_i,Z_i\}_{i=1}^5\cup\{Z_iZ_j\}_{1\le i<j\le 5}.
\end{equation}
Unlike the distance-$2$ case above, this error set includes all two-body $Z$-type correlators, so the KL conditions simultaneously constrain single-qubit and two-qubit $Z$-parity compressions. The unrestricted spectrum is $\Sigma_2(\boldsymbol{E})=[0,\sqrt3]$, realized by two complementary code families. 

For $\lambda^{*2}\in[0,5/2]$, a cyclic-invariant construction of the form
\begin{equation}
\begin{aligned}
|0_L\rangle &= a_0|00000\rangle
 + \frac{a_1}{\sqrt{5}}\sum_{\mathrm{cyc}}|00011\rangle\\
&\quad + \frac{a_2}{\sqrt{5}}\sum_{\mathrm{cyc}}|00101\rangle
 + \frac{a_3}{\sqrt{5}}\sum_{\mathrm{cyc}}|01111\rangle, \\
|1_L\rangle &= X^{\otimes 5}|0_L\rangle,
\end{aligned}
\label{eq:52_r2_cyc_ansatz}
\end{equation}
with $a_0,a_1,a_2\in\mathbb{R}$ and $a_3\in\mathbb{C}$ satisfies the full set of KL conditions. The associated projector is $P=|0_L\rangle\langle 0_L|+|1_L\rangle\langle 1_L|$. Solving the resulting constraint equations reduces all coefficients to a single free parameter $a_0$:
\begin{equation}
\begin{aligned}
a_1&=-a_2=\pm\sqrt{\frac{3}{8}-a_0^2},\\
|a_3|^2&=a_0^2+\frac{1}{4},\\
\Re(a_3)&=-\frac{1}{\sqrt{5}\,a_0}\Bigl(\frac{3}{8}-a_0^2\Bigr),
\end{aligned}
\end{equation}
and the signature norm takes the closed form
\begin{equation}
\lambda^{*2}=\frac{1}{10}\bigl(16a_0^2-1\bigr)^2.
\label{eq:52_r2_cyc_lambda}
\end{equation}
Here $a_0^2\in[1/16,\,3/8]$, and as $a_0$ varies over this range $\lambda^{*2}\in[0,\,5/2]$, so $\lambda^*\in[0,\sqrt{5/2}\,]$.

For $\lambda^{*2}\in[2,3]$, the remaining portion of the unrestricted spectrum is covered by a non-cyclic family.  Let $t\in[0,1/4]$ be a free parameter, and define the codewords
\begin{equation}
\begin{aligned}
|0_L(t)\rangle &= b_0|00110\rangle+b_1|01001\rangle+b_2|01100\rangle\\
&\quad +b_3|10001\rangle+b_4|10100\rangle+b_5|11011\rangle, \\
|1_L(t)\rangle &= Z_5\,X^{\otimes 5}|0_L(t)\rangle,
\end{aligned}
\label{eq:52_r2_noncyc}
\end{equation}
with real coefficients satisfying
\begin{equation}
\begin{aligned}
b_0^2&=b_5^2=\frac14,\\ 
b_1^2&=b_4^2=t,\\
b_2^2&=b_3^2=\frac14-t.
\end{aligned}
\label{eq:52_r2_noncyc_coeffs}
\end{equation}

Let $P(t)=|0_L(t)\rangle\langle 0_L(t)|+|1_L(t)\rangle\langle 1_L(t)|$. One verifies that $P(t)$ satisfies the KL conditions, and the signature norm evaluates to
\begin{equation}
\lambda^{*2}=1+2(1-4t)^2+2(4t)^2,
\end{equation}
which sweeps $\lambda^{*2}\in[2,3]$ as $t$ ranges over $[0,1/4]$.

Because the cyclic-invariant family \eqref{eq:52_r2_cyc_ansatz} already fills the interval $[0,\sqrt{5/2}\,]$, imposing \emph{projector-level} cyclic symmetry produces no further expansion: the numerical search over cyclic-symmetric projectors recovers the same interval. Thus both cyclic symmetry levels coincide:
\begin{equation}
\Sigma^{\mathrm{basis}}_{2,C_5}(\boldsymbol{E})
=\Sigma^{\mathrm{proj}}_{2,C_5}(\boldsymbol{E})
=\left[0,\sqrt{\frac{5}{2}}\right].
\end{equation}
The effect of replacing cyclic symmetry by full permutation invariance on the same error family $\mathcal{E}^{\mathrm{asym}}_{5,2}$ is analysed in Sec.~\ref{sec:52_asym_r2_perm}, where the spectrum acquires a strictly positive lower endpoint.

\subsubsection{A closed interval for $((5,2))$ code with $\mathcal{E}^{\mathrm{mix}}_5$}
\label{sec:52_full55_cyclic}

We next consider the symmetry-compatible detectable family
\begin{equation}
\begin{aligned}
\boldsymbol{E}:=\mathcal E^{\mathrm{mix}}_{5}
=
&\{X_i,Y_i,Z_i\}_{i=1}^5\\
&\cup
\{X_iX_j,\ Z_iZ_j,\ X_iZ_j,\ Z_iX_j\}_{1\le i<j\le 5},
\end{aligned}
\label{eq:52_full55_errors}
\end{equation}
which consists of $55$ distinct Hermitian Pauli observables.  For this family there is no distance parameter $d$; rather, the detectable set is specified directly.  Unrestricted Stiefel-manifold optimization over $St(32,2)$ gives the endpoint values
\begin{equation}
\lambda^{*2}_{\min}=0,
\qquad
\lambda^{*2}_{\max}\approx \frac54,
\end{equation}
with KL residuals at machine precision.  More importantly, the entire interval is already realized analytically inside the trivial cyclic sector \(\mathcal H_0\), so the cyclic restriction does not shrink the numerically observed range.

Let
\begin{equation}
|v\rangle_{\mathrm{cyc}}
:=
\frac{1}{\sqrt5}\sum_{r=0}^{4}T^r|v\rangle
\end{equation}
denote the normalized orbit sum under the five-cycle \(T\).  Inside the cyclic \(+1\) sector \(\mathcal H_0\), a convenient orthonormal basis is
\begin{equation}
\label{eq:52_full55_H0_basis}
\begin{aligned}
|e_k\rangle &= |D_{5,k-1}\rangle, \qquad k=1,\ldots,6,\\
|f_1\rangle &= \frac{|00011\rangle_{\mathrm{cyc}}-|00101\rangle_{\mathrm{cyc}}}{\sqrt2},\\
|f_2\rangle &= \frac{|00111\rangle_{\mathrm{cyc}}-|01011\rangle_{\mathrm{cyc}}}{\sqrt2},
\end{aligned}
\end{equation}
where the cyclic Dicke states are
\begin{equation}
\begin{aligned}
|D_{5,0}\rangle &= |00000\rangle,\\
|D_{5,1}\rangle &= |00001\rangle_{\mathrm{cyc}},\\
|D_{5,2}\rangle &= \frac{|00011\rangle_{\mathrm{cyc}}+|00101\rangle_{\mathrm{cyc}}}{\sqrt2},\\
|D_{5,3}\rangle &= \frac{|00111\rangle_{\mathrm{cyc}}+|01011\rangle_{\mathrm{cyc}}}{\sqrt2},\\
|D_{5,4}\rangle &= |01111\rangle_{\mathrm{cyc}},\\
|D_{5,5}\rangle &= |11111\rangle.
\end{aligned}
\end{equation}
Thus
\begin{equation}
\begin{aligned}
\mathcal H_0 &= V_{5/2}\oplus V_{1/2},\\
V_{5/2} &= \mathrm{span}\{|e_1\rangle,\dots,|e_6\rangle\},\\
V_{1/2} &= \mathrm{span}\{|f_1\rangle,|f_2\rangle\}.
\end{aligned}
\end{equation}
The global bit flip \(X^{\otimes 5}\) exchanges the pairs
\begin{equation}
\begin{aligned}
|e_1\rangle &\leftrightarrow |e_6\rangle,\qquad 
|e_5\rangle \leftrightarrow |e_2\rangle,\\
|e_3\rangle &\leftrightarrow |e_4\rangle,\qquad
|f_1\rangle \leftrightarrow |f_2\rangle.
\end{aligned}
\end{equation}
Representation-theoretically, the nontrivial interpolation is produced by mixing the symmetric block \(V_{5/2}\) with the extra \(V_{1/2}\) block inside \(\mathcal H_0\).

An exact one-parameter family is obtained as follows.  For \(x\in[0,1]\), define
\begin{equation}
\begin{aligned}
\alpha(x) &=
-\frac{\sqrt{2+x}}{4\sqrt2}
+i\,\frac{\sqrt{x}}{2\sqrt2},\\[1mm]
\beta(x) &=
\frac{\sqrt5}{4\sqrt2}\sqrt{2+x},\\[1mm]
\gamma(x) &=
-\,i\,\frac{\sqrt5}{4}\sqrt{x},\\[1mm]
\delta(x) &=
\frac{\sqrt5}{4}\sqrt{\frac{1-x}{1+x}}
\bigl(\sqrt{2+x}+i\sqrt{x}\bigr),
\end{aligned}
\label{eq:52_full55_coeffs}
\end{equation}
and set
\begin{equation}
\begin{aligned}
|0_L(x)\rangle
&=
\alpha(x)|e_1\rangle
+\beta(x)|e_5\rangle
+\gamma(x)|e_3\rangle
+\delta(x)|f_1\rangle,\\
|1_L(x)\rangle
&=
X^{\otimes 5}|0_L(x)\rangle\\
&=
\alpha(x)|e_6\rangle
+\beta(x)|e_2\rangle
+\gamma(x)|e_4\rangle
+\delta(x)|f_2\rangle.
\end{aligned}
\label{eq:52_full55_family}
\end{equation}
A direct substitution shows that \(\{|0_L(x)\rangle,|1_L(x)\rangle\}\) is orthonormal for every \(x\in[0,1]\), and that the corresponding projector
\begin{equation}
P(x)=|0_L(x)\rangle\langle 0_L(x)|+|1_L(x)\rangle\langle 1_L(x)|
\end{equation}
satisfies the full KL compression equations
\begin{equation}
P(x)\,O\,P(x)=\lambda_O(x)\,P(x),
\qquad
O\in\mathcal E.
\end{equation}

Along this family all one-body terms and all mixed \(XZ/ZX\) terms vanish:
\begin{equation}
\begin{aligned}
\lambda_{X_i}(x) &= \lambda_{Y_i}(x)=\lambda_{Z_i}(x)=0,\\
\lambda_{X_iZ_j}(x) &= \lambda_{Z_iX_j}(x)=0.
\end{aligned}
\end{equation}
Hence only the \(XX\) and \(ZZ\) coefficients contribute.  Writing
\begin{equation}
\eta(x):=\sqrt{\frac{1-x}{1+x}},
\end{equation}
one finds that the nonzero coefficients depend only on distance on the $5$-cycle:
\begin{equation}
\begin{aligned}
\lambda_{X_iX_{i+1}}(x)
&=
\lambda_{Z_iZ_{i+2}}(x)
=
\frac{x}{4}\bigl(1+\eta(x)\bigr),\\
\lambda_{X_iX_{i+2}}(x)
&=
\lambda_{Z_iZ_{i+1}}(x)
=
\frac{x}{4}\bigl(1-\eta(x)\bigr),
\end{aligned}
\label{eq:52_full55_lambda_coeffs}
\end{equation}
for all \(i\in\mathbb Z_5\), with indices understood modulo \(5\). Therefore
\begin{equation}
\lambda^{*2}(x)
=
5\sum_{r=1}^{2}
\Bigl(
\lambda_{X_1X_{1+r}}(x)^2
+
\lambda_{Z_1Z_{1+r}}(x)^2
\Bigr)
=
\frac{5}{2}\frac{x^2}{1+x}.
\label{eq:52_full55_lambda}
\end{equation}
This is strictly increasing on \([0,1]\), so the family fills the entire interval
\begin{equation}
\lambda^{*2}\in\left[0,\frac54\right],
\qquad
\lambda^*\in\left[0,\frac{\sqrt5}{2}\right].
\label{eq:52_full55_interval}
\end{equation}
Equivalently, for any target \(t\in[0,5/4]\), choosing
\begin{equation}
x=\frac{t+\sqrt{t(t+10)}}{5}
\end{equation}
gives an exact cyclic code with \(\lambda^{*2}=t\).

The endpoints have a transparent interpretation.  At \(x=0\),
\begin{equation}
\begin{aligned}
|0_L(0)\rangle
&=
-\frac14|e_1\rangle
+\frac{\sqrt5}{4}|e_5\rangle
+\frac{\sqrt{10}}{4}|f_1\rangle,\\
|1_L(0)\rangle
&=
X^{\otimes 5}|0_L(0)\rangle,
\end{aligned}
\end{equation}
which is the familiar \(\lambda^*=0\) endpoint.  At \(x=1\), one has \(\delta(1)=0\), so the \(V_{1/2}\) component disappears and the family lies entirely inside the symmetric block \(V_{5/2}\subset\mathcal H_0\).  In that case
\begin{equation}
\begin{aligned}
\lambda_{X_iX_j} &= \lambda_{Z_iZ_j} = \frac14,\\
\lambda_{X_i} &= \lambda_{Y_i} = \lambda_{Z_i} =
\lambda_{X_iZ_j}
=
\lambda_{Z_iX_j}
= 0,
\end{aligned}
\end{equation}
hence
\begin{equation}
\lambda^{*2}=10\left(\frac14\right)^2+10\left(\frac14\right)^2=\frac54.
\end{equation}

Consequently,
\begin{equation}
\left[0,\frac{\sqrt5}{2}\right]
\subseteq
\Sigma^{\mathrm{basis}}_{2,C_5}(\boldsymbol E)
\subseteq
\Sigma^{\mathrm{proj}}_{2,C_5}(\boldsymbol E)
\subseteq
\Sigma_2(\boldsymbol E),
\end{equation}
and our numerical searches over both \(St(8,2)\) inside \(\mathcal H_0\) and unrestricted \(St(32,2)\) recover the same endpoint values, so no further enlargement is observed beyond \eqref{eq:52_full55_interval}. The permutation-invariant restriction of this family is analysed in Sec.~\ref{sec:52_full55_perm}.

\subsection{Permutation invariance}
Under full permutation symmetry $S_n$ the state--projector gap is especially pronounced.
State-level permutation invariance confines each basis vector to the fully symmetric subspace, which is often too small to support a code of the required dimension, whereas requiring only projector-level invariance allows the code space to carry a nontrivial $S_n$-representation and can thereby accommodate exact error detection.
We first exhibit this existence gap for the single-qubit error set $\mathcal{E}=\{X_i,Y_i,Z_i\}$ across four parameter sets, then analyse the asymmetric family $\mathcal{E}^{\mathrm{asym}}_{5,2}$ of Sec.~\ref{sec:52_asym_r2_cyclic}, and finally examine the mixed family $\mathcal{E}^{\mathrm{mix}}_5$ of Sec.~\ref{sec:52_full55_cyclic}.

\subsubsection{State--projector existence gap for $\mathcal{E}=\{X_i,Y_i,Z_i\}$.}
\label{sec:PI_single_qubit}
For the parameter sets $((4,3,2))$, $((4,4,2))$, $((5,4,2))$, and $((5,5,2))$, state-level permutation invariance forces each basis vector into the Dicke subspace, which is too restrictive to support $K$ mutually orthonormal vectors satisfying the KL conditions; hence $\Sigma^{\mathrm{basis}}_{K,S_n}(\boldsymbol{E})=\emptyset$ in each case.
At the projector level, however, explicit constructions yield permutation-invariant projectors with exact KL error detection for each of these four parameter sets, realising $0\in\Sigma^{\mathrm{proj}}_{K,S_n}(\boldsymbol{E})$ in every case.

For $((4,3,2))$, one may span a $3$-dimensional code subspace by
\begin{equation}
\begin{aligned}
|0_L\rangle &= \frac{1}{\sqrt{2}}\bigl(|0101\rangle-|1010\rangle\bigr), \\
|1_L\rangle &= \frac{1}{\sqrt{2}}\bigl(|0110\rangle-|1001\rangle\bigr), \\
|2_L\rangle &= \frac{1}{\sqrt{2}}\bigl(|0011\rangle-|1100\rangle\bigr),
\end{aligned}
\end{equation}
and set $P=\sum_{j=0}^2 |j_L\rangle\langle j_L|$.  Although the individual vectors above are not symmetric under all permutations, the span is invariant, and one checks $U_\pi P U_\pi^\dagger=P$ for all $\pi\in S_4$, with $P E P=0$ for all $E\in\{X_i,Y_i,Z_i\}_{i=1}^4$.

For $((4,4,2))$, an exact permutation-invariant projector at $\lambda^*(P)=0$ is
\begin{equation}
P=\frac{1}{4}(1+XXXX)(1+ZZZZ).
\end{equation}

For $((5,4,2))$, an explicit projector-level permutation-invariant code at $\lambda^*(P)=0$ is obtained by taking $P=\sum_{j=0}^3|j_L\rangle\langle j_L|$, where
\begin{widetext}
\begin{equation}
\begin{aligned}
|0_L\rangle
&= \frac{1}{2}\bigl(
   |00001\rangle+|11110\rangle
 - |00010\rangle-|11101\rangle
 \bigr), \\
|1_L\rangle
&= \frac{1}{\sqrt{12}}\bigl(
   |00001\rangle+|11110\rangle
 + |00010\rangle+|11101\rangle 
 - 2|00100\rangle-2|11011\rangle
 \bigr), \\
|2_L\rangle
&= \frac{1}{\sqrt{24}}\bigl(
   |00001\rangle+|11110\rangle
 + |00010\rangle+|11101\rangle 
 + |00100\rangle+|11011\rangle
 - 3|01000\rangle-3|10111\rangle
 \bigr), \\
|3_L\rangle
&= \frac{1}{\sqrt{40}}\bigl(
   |00001\rangle+|11110\rangle
 + |00010\rangle+|11101\rangle 
 + |00100\rangle+|11011\rangle
 + |01000\rangle+|10111\rangle 
 - 4|10000\rangle-4|01111\rangle
\bigr).
\end{aligned}
\end{equation}
\end{widetext}

This projector satisfies exact KL for $\{X_i,Y_i,Z_i\}_{i=1}^5$, hence $\lambda^*(P)=0$.

Lastly, for $((5,5,2))$, a rank-$5$ permutation-invariant projector at $\lambda^*(P)=0$ is given by $P=\sum_{j=0}^4|j_L\rangle\langle j_L|$, where
\begin{equation}
\begin{aligned}
|0_L\rangle &= \frac{1}{\sqrt{2}}\bigl(|00001\rangle \pm |11110\rangle\bigr),\\
|1_L\rangle &= \frac{1}{\sqrt{2}}\bigl(|01111\rangle \pm |10000\rangle\bigr),\\
|2_L\rangle &= \frac{1}{\sqrt{2}}\bigl(|01000\rangle \pm |10111\rangle\bigr),\\
|3_L\rangle &= \frac{1}{\sqrt{2}}\bigl(|11011\rangle \pm |00100\rangle\bigr),\\
|4_L\rangle &= \frac{1}{\sqrt{2}}\bigl(|00010\rangle \pm |11101\rangle\bigr).
\end{aligned}
\end{equation}
Here both $\pm$ signs must be chosen uniformly (all $+$ or all $-$) for the projector to remain permutation-invariant; the KL conditions themselves are sign-independent. One can verify that exact KL detection holds for all $E \in \{X_i, Y_i, Z_i\}$, and hence $\lambda^*(P) = 0$.

These examples confirm that while demanding a basis of individually permutation-invariant codewords can be overly restrictive, requiring only projector-level symmetry retains full permutation invariance of the code space and admits exact error-detecting codes even when no state-level permutation-invariant basis is available. In terms of spectra, the mechanism is
\begin{equation}
\Sigma^{\mathrm{basis}}_{K,S_n}(\boldsymbol{E})\subsetneq\Sigma^{\mathrm{proj}}_{K,S_n}(\boldsymbol{E})
\end{equation}
in the existence-gap instances listed above, with $\Sigma^{\mathrm{basis}}_{K,S_n}(\boldsymbol{E})=\emptyset$ but $0\in\Sigma^{\mathrm{proj}}_{K,S_n}(\boldsymbol{E})$.

\subsubsection{Asymmetric permutation-invariant codes for $\mathcal E^{\mathrm{asym}}_{5,2}$.}
\label{sec:52_asym_r2_perm}
For the same asymmetric family $\boldsymbol{E}=\mathcal E^{\mathrm{asym}}_{5,2}$ studied under cyclic symmetry in Sec.~\ref{sec:52_asym_r2_cyclic}, both state-level and projector-level permutation invariance lead to the same signature spectrum. In contrast with the cyclic case, this permutation-invariant spectrum has a strictly positive lower endpoint.

A state-level permutation-invariant code is supported on the fully symmetric Dicke sector. Imposing the logical bit-flip relation \( |1_L\rangle = X^{\otimes 5}|0_L\rangle\), we use the expansion
\begin{equation}
\begin{aligned}
|0_L\rangle &= c_0\,|D_{5,0}\rangle + c_1\,|D_{5,2}\rangle + c_2\,|D_{5,4}\rangle, \\
|1_L\rangle &= X^{\otimes 5}|0_L\rangle,
\end{aligned}
\label{eq:52_r2_PI_ansatz}
\end{equation}
and define the associated code projector by $P=|0_L\rangle\langle 0_L|+|1_L\rangle\langle 1_L|$. Choosing the global phase such that \(c_0\ge 0\) is real, the normalization condition together with the single-qubit \(Z_i\) KL equation gives
\begin{equation}
|c_1|^2=\frac34-2c_0^2,\qquad
|c_2|^2=c_0^2+\frac14,
\end{equation}
so in particular \(c_0^2\le 3/8\). As in the cyclic family of Sec.~\ref{sec:52_asym_r2_cyclic}, the signature norm is
\begin{equation}
\lambda^{*2}=\frac{2}{5}\Bigl(8c_0^2-\frac{1}{2}\Bigr)^2.
\label{eq:52_r2_PI_lambda}
\end{equation}
The difference is that, in the permutation-invariant setting, the remaining KL equations for \(X_i\) and \(Y_i\) impose the additional solvability condition
\begin{equation}
256c_0^4-512c_0^2+81\le 0,
\end{equation}
hence
\begin{equation}
c_0^2\in\left[1-\frac{5\sqrt7}{16},\,\frac38\right].
\label{eq:52_r2_PI_c0_range}
\end{equation}
Substituting \eqref{eq:52_r2_PI_c0_range} into \eqref{eq:52_r2_PI_lambda} yields
\begin{equation}
\lambda^{*2}\in\left[40-15\sqrt7,\;\frac52\right]
\approx [\,0.3137,\;2.5\,],
\end{equation}
and therefore
\begin{equation}
\lambda^*\in\left[\sqrt{40-15\sqrt7},\;\sqrt{\frac52}\right]
\approx [\,0.5601,\;1.5811\,].
\label{eq:52_r2_PI_range}
\end{equation}
The lower endpoint is attained at \(c_0^2=1-5\sqrt7/16\), whereas the upper endpoint is attained at \(c_0^2=3/8\). Although \eqref{eq:52_r2_PI_lambda} would vanish at \(c_0^2=1/16\), that value lies outside the feasible interval \eqref{eq:52_r2_PI_c0_range}; hence the permutation-invariant spectrum is bounded away from zero.

For \(n=5\) and \(K=2\), projector-level permutation invariance does not enlarge this interval. Indeed, any rank-2 projector commuting with the full \(S_5\) action must be supported entirely on the trivial irreducible representation, since the nontrivial \(S_5\)-sectors contribute ranks in multiples of \(4\) or \(5\). Thus projector-level permutation invariance reduces to the same Dicke-sector problem as the state-level expansion above, and
\begin{equation}
\Sigma^{\mathrm{basis}}_{2,S_5}(\boldsymbol{E})
=
\Sigma^{\mathrm{proj}}_{2,S_5}(\boldsymbol{E})
=
\left[\sqrt{40-15\sqrt7},\;\sqrt{\frac52}\right].
\end{equation}

\subsubsection{Permutation-invariant codes for $\mathcal{E}^{\mathrm{mix}}_5$.}
\label{sec:52_full55_perm}

We now examine the mixed error family $\boldsymbol{E}=\mathcal E^{\mathrm{mix}}_5$ of \eqref{eq:52_full55_errors}, whose cyclic-invariant spectrum was determined in Sec.~\ref{sec:52_full55_cyclic}, under full permutation invariance. Restricting the code projector to the fully permutation-invariant subspace, the problem reduces to the Dicke (spin-$5/2$) irreducible representation
\begin{equation}
\mathcal H_{\mathrm{sym}}
=
\mathrm{span}\{|D_{5,s}\rangle\}_{s=0}^{5}
\cong V_{5/2},
\qquad
\dim \mathcal H_{\mathrm{sym}}=6.
\end{equation}
On \(\mathcal H_{\mathrm{sym}}\), the $55$ operators in \eqref{eq:52_full55_errors} collapse, by permutation symmetry, to six distinct restrictions:
$X_i \sim X_1,
Y_i \sim Y_1,
Z_i \sim Z_1,
X_iX_j \sim X_1X_2,
Z_iZ_j \sim Z_1Z_2,
X_iZ_j \sim Z_iX_j \sim X_1Z_2.$
with multiplicities \(5,5,5,10,10,20\), respectively. Hence for any rank-$2$ projector \(P=VV^\dagger\) with \(\mathrm{Ran}(P)\subseteq \mathcal H_{\mathrm{sym}}\),
\begin{equation}
\lambda^{*2}
=
5(\lambda_X^2+\lambda_Y^2+\lambda_Z^2)
+
10(\lambda_{XX}^2+\lambda_{ZZ}^2)
+
20\lambda_{XZ}^2,
\end{equation}
where
\begin{equation}
\begin{aligned}
V^\dagger X_1V&=\lambda_X I_2, \\
V^\dagger Y_1V&=\lambda_Y I_2, \\
V^\dagger Z_1V&=\lambda_Z I_2, \\
V^\dagger X_1X_2V&=\lambda_{XX} I_2, \\
V^\dagger Z_1Z_2V&=\lambda_{ZZ} I_2, \\
V^\dagger X_1Z_2V&=\lambda_{XZ} I_2.
\end{aligned}
\end{equation}

Direct optimization over \(St(6,2)\), with the KL loss built from these six representatives and their multiplicities, did not reveal any interior exact solutions. In \(200\) random initializations, \(194\) converged to machine-precision KL residuals \(\lesssim 10^{-14}\), and within numerical precision,
\begin{equation}
\lambda^{*2}_{\mathrm{sym}}=\frac54.
\end{equation}
Moreover, when we targeted interior values \(t\in(0,5/4)\) by minimizing
\begin{equation}
\mu L_{\mathrm{KL}}+(\lambda^{*2}-t)^2,
\end{equation}
the optimizer always returned to \(\lambda^{*2}=5/4\) once \(L_{\mathrm{KL}}\) was forced small.

The endpoint \(x=1\) of the cyclic family \eqref{eq:52_full55_family} already lies in \(\mathcal H_{\mathrm{sym}}\), so it gives an explicit permutation-invariant code at \(\lambda^{*2}=5/4\). A convenient closed form is obtained by setting
\begin{equation}
\tau:=\frac{2+i\sqrt3}{\sqrt7},
\qquad |\tau|=1,
\end{equation}
and defining
\begin{equation}
\begin{aligned}
|0_L^{\mathrm{PI}}\rangle
&=
\frac{1}{\sqrt{32}}
\Big(
\sqrt7\,|D_{5,0}\rangle
-\sqrt{10}\,\tau\,|D_{5,2}\rangle
+i\sqrt{15}\,\tau\,|D_{5,4}\rangle
\Big),\\[4pt]
|1_L^{\mathrm{PI}}\rangle
&=
X^{\otimes 5}|0_L^{\mathrm{PI}}\rangle \\
&=
\frac{1}{\sqrt{32}}
\Big(
\sqrt7\,|D_{5,5}\rangle
-\sqrt{10}\,\tau\,|D_{5,3}\rangle
+i\sqrt{15}\,\tau\,|D_{5,1}\rangle
\Big).
\end{aligned}
\end{equation}
The corresponding projector
\begin{equation}
P_{\mathrm{PI}}
=
|0_L^{\mathrm{PI}}\rangle\langle 0_L^{\mathrm{PI}}|
+
|1_L^{\mathrm{PI}}\rangle\langle 1_L^{\mathrm{PI}}|
\end{equation}
satisfies all KL conditions, with
\begin{equation}
\begin{aligned}
\lambda_{X_iX_j} &= \lambda_{Z_iZ_j}=\frac14,\\
\lambda_{X_i} &= \lambda_{Y_i}=\lambda_{Z_i}=\lambda_{X_iZ_j} = \lambda_{Z_iX_j}=0,
\end{aligned}
\end{equation}
for all \(1\le i<j\le 5\). Therefore
\begin{equation}
\lambda^{*2}(P_{\mathrm{PI}})
=
10\left(\frac14\right)^2+10\left(\frac14\right)^2
=
\frac54.
\end{equation}

Thus the upper endpoint of \eqref{eq:52_full55_interval} is already realized inside the fully permutation-invariant sector. By contrast, our numerical search inside \(\mathcal H_{\mathrm{sym}}\) suggests that no interior exact points occur there. In other words, the full interval \([0,5/4]\) arises only after enlarging from the symmetric sector \(V_{5/2}\) to the cyclic \(T=+1\) sector
\begin{equation}
\mathcal H_0\cong V_{5/2}\oplus V_{1/2},
\end{equation}
where the extra \(V_{1/2}\) block supplies the missing degree of freedom needed for the interpolation.


\section{Discussion}
\label{sec:discussion}

This paper studies exact Pauli-detecting quantum codes through the scalar signature norm $\lambda^*$, obtained from the simultaneously compressed Pauli expectations of a rank-$K$ detecting projector. From the viewpoint of joint higher-rank numerical ranges and Knill--Laflamme compression conditions \cite{KnillLaflamme,li2008generalized,ChoiKribsZyczkowski2006LAA,LiPoonSze2008JMAA,GauLiPoonSze2011,Woerdeman2008HigherRankConvex,afshin2018pauli}, the problem begins as a vector-valued feasibility question for simultaneous scalar compressions. Passing to $\lambda^*$ reduces this to a one-dimensional invariant that still retains a direct variance interpretation for Pauli observables \cite{XuSchwonnekWinter2024,xu2025simultaneous}.

The results support a refined interval picture. For unrestricted detection problems, the attainable signature spectrum $\Sigma_K(\boldsymbol{E})$ was a single closed interval whenever non-empty in every example we studied, including the complete two-qubit classification and the random three-qubit scans. The same interval behavior persisted in all symmetry-restricted examples arising from symmetry-compatible pairs $(\mathcal E,G)$, where cyclic or permutation symmetry acts by orbit reduction rather than by an external slicing constraint. In those settings, the symmetry-restricted spectra may shrink, collapse to singletons, or become empty, but when non-empty they again appear as closed intervals.

The three-qubit random-tuple analysis also shows that this regularity has a genuine boundary. If one imposes a cyclic-$+1$ constraint on a Pauli tuple that is not itself stable under cyclic relabeling, then the symmetry condition is external to the error model. In that regime we observed empty spectra, singleton spectra, genuine intervals, and an exact disconnected example with spectrum $\{0,1\}$. Thus the relevant structural distinction is not simply ``with symmetry'' versus ``without symmetry'', but symmetry-compatible reduction versus externally imposed symmetry restriction. A similar kind of disconnectedness appears when one imposes additional transversal-gate symmetry constraints: in the Stiefel-manifold search for nonadditive codes with prescribed transversal groups in Ref.~\cite{zhang2025transversal}, Table~3 reports disconnected $\lambda^*$-ranges for certain groups (e.g.\ $2I$ at $\{0,\sqrt{3/4},\sqrt{7}\}\cup[1.34,1.52]$ and $BD_{10}$ at $[0,2.52]\cup\{\sqrt{7}\}$), and emphasizes that ``the listed intervals are not exhaustive but represent all ranges discovered so far'' \cite{zhang2025transversal}. In that sense, the disconnected cyclic-$+1$ spectra here and the disconnected transversal $\lambda^*$-regions there are in a similar spirit: they both reflect how an externally imposed structural constraint can slice an otherwise interval-like feasibility set into multiple components.

At the same time, the interval picture reveals a striking contrast between stabilizer and nonadditive codes. For stabilizer projectors, the Pauli compression coefficients lie in $\{0,\pm1\}$, so the corresponding $\lambda^*$-values form a discrete subset of the ambient spectrum. Consequently, whenever a nontrivial interval occurs, it contains a continuum of exact codes that cannot be exhausted by stabilizer constructions alone. From this perspective, stabilizer codes appear not as the generic points of the landscape but as distinguished discrete reference points inside a much larger nonadditive geometry \cite{Rains1997Nonadditive,Cross2009CWS,Laflamme1996FiveQubit}. One conceptual contribution of the present work is therefore to place stabilizer, nonadditive, symmetric, and asymmetric exact codes inside a common higher-rank variance framework, where the difference between these families is reflected directly in the geometry of their attainable signature spectra.

A second structural theme is the distinction between state-level and projector-level symmetry. In the symmetry-compatible families studied here, relaxing from individually symmetric codewords to a symmetric projector can enlarge an interval, replace a singleton by a continuum, or recover exact detecting codes when no symmetric basis exists. The larger cyclic and permutation examples show that the projector, rather than any preferred basis, is the natural symmetry carrier for exact detection \cite{beigi_et_al:LIPIcs.TQC.2013.192,Ouyang2014PermutationInvariant,Ouyang2016PermutationInvariantMultiQubit}.

Methodologically, the paper combines exact low-dimensional analysis with Stiefel-manifold optimization and symmetry-adapted parameterizations. The numerical component is therefore not a substitute for structure, but a way of probing regimes where the compression equations are too large to solve analytically while still preserving the exact Knill--Laflamme constraints at the level of the objective. At present, the interval phenomenon should be viewed as strong evidence rather than as a theorem. Natural open questions include establishing general conditions under which $\Sigma_K(\boldsymbol{E})$ must be an interval in the unrestricted and symmetry-compatible settings, characterizing when external symmetry constraints can destroy connectedness, understanding when the full higher-rank variance range $Q_K(\boldsymbol{E})$ is connected, and clarifying how symmetry sectors control the endpoints of $\Sigma_K(\boldsymbol{E})$. More broadly, the present results suggest that higher-rank Pauli variance geometry is substantially more ordered than the full signature-vector picture alone would indicate, and that $\lambda^*$ provides an effective coordinate for navigating the exact-code landscape across stabilizer, nonadditive, cyclic, and permutation-restricted families.

\begin{acknowledgments}

We thank Chao Zhang, Ningping Cao and Yiu-Tung Poon for helpful discussions.
	
\end{acknowledgments}




\bibliography{refs}

\end{document}